\documentclass[10pt,a4paper]{article}

\usepackage{color}
\usepackage[latin1]{inputenc}
\usepackage{amsmath}
\usepackage{amsfonts}
\usepackage{amssymb}
\usepackage{appendix}

\usepackage{graphicx}
\usepackage[pdftex,bookmarks]{hyperref}%%indice dinamico

\usepackage{amssymb}
\usepackage{amsthm}

\newtheorem{theorem}{Theorem}[section]

\newtheorem{lemma}[theorem]{Lemma}
\newtheorem{proposition}[theorem]{Proposition}
\theoremstyle{definition}
\newtheorem{definition}[theorem]{Definition}
\newtheorem{remark}{Remark}

\def\R{\mathbb{R}}
\let\e=\varepsilon
\let\var=\varepsilon

\let\.=\cdot
\let\0=\emptyset

\def\pa{\partial}

\def\square{\hbox{$\sqcap\kern-7pt\sqcup$}}

\def\dev{\operatorname{div}}
\newcommand{\dsp}{\displaystyle}
\newcommand{\ints}{\int_{\R^3}}
\newcommand{\intd}{\iint_{\R^3\times \R^3}}

\title{Polynomial propagation of moments and global existence for a Vlasov-Poisson system with a point charge}
\author{Laurent Desvillettes, Evelyne Miot, Chiara Saffirio}

\def\adresse{
\begin{description}
\item[L.  Desvillettes:] CMLA, ENS Cachan \& CNRS,
 61 Av. du Pdt. Wilson, 94235 Cachan Cedex, France.
Email: \texttt{desville@cmla.ens-cachan.fr}

\item[E. Miot:] Universit\'e Paris-Sud, D\'epartement de math\'ematiques, B\^atiment 425, 91405 Orsay, France.
Email: \texttt{Evelyne.Miot@math.u-psud.fr}
\item[C. Saffirio:] Institute for Applied Mathematics, University of Bonn, 53115 Bonn, Germany
Email: \texttt{chiara.saffirio@hausdorff-center.uni-bonn.de}
\end{description}
}

\date{\today}

\begin{document}

\maketitle

\begin{abstract}
In this paper, we extend to the case of initial data constituted of a Dirac mass plus a bounded density (with finite moments) the theory of Lions and Perthame \cite{LP} for the Vlasov-Poisson equation. Our techniques also provide polynomially growing in time estimates for moments of the bounded density.
\end{abstract}

\section{Introduction}

The aim of this article is to propose new estimates for moments of the three--dimensional Vlasov--Poisson system, which enable us to improve the existing theory in the following directions: first the possibility of considering initial data which are the sum of a Dirac mass and a continuous density which does not vanish in the vicinity of the Dirac mass; secondly, the opportunity of bounding polynomially (in an explicit computable way) with respect to time the high order moments of the distribution (when the initial datum is not assumed to be compactly supported).   
The three-dimensional Vlasov-Poisson system  for initial data containing one Dirac mass writes 
\begin{equation}
\label{syst:VP}\begin{cases}
\dsp  \partial_t f+v\cdot \nabla_x f+(E+F)\cdot \nabla_v f=0 
\vspace*{0.5em}\\
\dsp E(t,x)=\int_{\R^3} \frac{x-y}{|x-y|^3}\rho(t,y)\,dy\vspace*{0.5em}\\
\dsp \rho(t,x)=\int_{\R^3} f(t,x,v)\,dv\vspace*{0.5em}\\
\dsp F(t,x)=\frac{x-\xi(t)}{|x-\xi(t)|^3}.\end{cases}
\end{equation}
Here  $f:=f(t,x,v)\in L^\infty(\R_+\times \R^3\times \R^3)$ corresponds to a nonnegative density of charged particles in a plasma, subjected to a self-induced electric force field $E:=E(t,x)$. The plasma interacts with a point charge, located at $\xi(t)$ with
velocity $\eta(t)$, which induces the singular electric field $F:=F(t,x)$. The evolution of the charge is itself given by 
 \begin{equation}\label{syst:VP-2}
\begin{cases}
\dsp \dot{\xi}(t)=\eta(t),\\
\dsp \dot{\eta}(t)=E(t,\xi(t)).
 \end{cases}
\end{equation}
The initial conditions associated to \eqref{syst:VP}-\eqref{syst:VP-2} are
 \begin{equation}\label{syst:VP-3}
 (\xi(0),\eta(0))=(\xi_0,\eta_0),\quad f(0,x,v)=f_0(x,v).
\end{equation}
When there is no point charge, then \eqref{syst:VP}-\eqref{syst:VP-2} reduces to \eqref{syst:VP} with $F=0$, which is the standard Vlasov-Poisson system. As announced, \eqref{syst:VP}-\eqref{syst:VP-2} can also be thought of as the standard Vlasov-Poisson system for the total density $f(t)+\delta_{\xi(t)}\otimes \delta_{\eta(t)}$. 

In this paper we will focus on weak solutions to \eqref{syst:VP}, which we define as couples $(f,\xi)$ such that the first equation in \eqref{syst:VP} is satisfied in the sense of distributions and such that \eqref{syst:VP-2} holds in the classical sense.
Our main global existence and propagation of moments result of such solutions with finite moments is stated in the following
\begin{theorem}
\label{thm:main}
There exists $0<\lambda\leq 1$ with the following property : If $f_0\in L^1\cap L^\infty(\R^3\times \R^3)$ is 
nonnegative,  $(\xi_{0},\eta_{0})\in \R^3 \times \R^3$, and
\begin{itemize}
\item[(i)] $\dsp \mathcal{M}_0 =\iint_{\R^3\times \R^3} f_0(x,v)\,dx\,dv<\lambda$;
\item[(ii)] There exists $m_0>6$ such that for all $ m< m_0$
\begin{equation*}
 \iint_{\R^3\times \R^3}  \left( |v|^2+\frac{1}{|x-\xi_{0}|}\right)^{m/2}\,f_0(x,v)\, dx\, dv<+\infty;
\end{equation*}
\end{itemize}

Then there exists a global weak solution $(f,\xi)$ to the system
\eqref{syst:VP}--\eqref{syst:VP-3}, with $f\in C(\R_+,L^p(\R^3\times \R^3))\cap L^\infty(\R_+,L^\infty(\R^3\times \R^3))$ for any $1\leq p<+\infty$, $E\in L^\infty([0,T],L^\infty(\R^3))$ for all $T>0$ and $\xi\in C^2(\R_+)$.

Moreover, for all $t\in \R_+$ and for all $m<\min(m_0,7)$,
%\begin{equation*}
%\iint_{\R^3\times \R^3}  \left( |v|^2+\frac{1}{|x-\xi(t)|}\right)^{m/2}
%\,f(t,x,v)\, dx\, dv<+\infty.
%\end{equation*}
%\end{theorem}
%\begin{remark}
%In fact one is able to obtain a polynomial in time growth on the moments:
\begin{equation*}
\iint_{\R^3\times \R^3}  \left( |v|^2+\frac{1}{|x-\xi(t)|^{}}\right)^{m/2}
\,f(t,x,v)\, dx\, dv\leq C\,(1+t)^c,
\end{equation*}
where $C$ and $c$ only depend on $(f_0,(\xi_0,\eta_0))$ and can be estimated explicitly (cf. Remark \ref{rk:explicit-const}). 
\end{theorem}

%\end{remark}

\medskip

\begin{remark}
Since the density $f$  constructed in Theorem \ref{thm:main} has bounded moments of
order higher than 6 and belongs to $C(\R_+,L^p(\R^3\times \R^3))$, the field  $E$ has some additional regularity: it belongs to $C(\R_+,C^{0,\alpha}(\R^3))$ for some $0<\alpha<1$ (see Corollary 2 in \cite{LP}). In particular the ODE \eqref{syst:VP-2} holds in the classical sense.

\end{remark}

 We stress that if we modified the assumption $m_0>6$ in Theorem \ref{thm:main} (for example if $3<m_0<6$), then the electric field $E$ might be not continuous and we would not be able to prove the existence of the trajectory $t\mapsto (\xi(t),\eta(t))$ for all initial datum $(\xi_0,\eta_0)$.

\medskip

Note however that if we considered the standard Vlasov--Poisson system (i.e. system \eqref{syst:VP} with $F=0$), we would not need $m_0>6$ (the assumption $m_0>3$ would be enough). We would also not need assumption $(i)$, so that we would be left with the assumptions of \cite{LP}. 
In that case, we would still get the explicit, polynomial with respect to time estimate of the moments of order $m< \min(m_0,7)$, which provides an improvement of the result in \cite{LP}, where the growth with respect to time of the moments is of the form $\exp(\exp(t))$.  

Note finally that if we  fix the point charge $\xi=\xi_0$, which corresponds to the evolution of a plasma density under the influence of a stationary exterior singular field:
\begin{equation}
\label{syst:VP-bis}\begin{cases}
\dsp  \partial_t f+v\cdot \nabla_x f+(E+F_*)\cdot \nabla_v f=0 
\vspace*{0.5em}\\
\dsp E(t,x)=\int_{\R^3} \frac{x-y}{|x-y|^3}\rho(t,y)\,dy\vspace*{0.5em}\\
\dsp \rho(t,x)=\int_{\R^3} f(t,x,v)\,dv\vspace*{0.5em}\\
\dsp F_*(t,x)=C\,\frac{x-\xi_0}{|x-\xi_0|^3},\end{cases}
\end{equation}
then in this situation also, conditions $(i)$ and $(ii)$ of Theorem \ref{thm:main} are not needed and we can obtain the following
\begin{proposition}
\label{thm:main-bis}
Let $f_0\in L^1\cap L^\infty(\R^3\times \R^3)$ be nonnegative, let
$\xi_{0}\in \R^3$ such that there exists $m_0>3$ for which for all $ m< m_0$
\begin{equation*}
 \iint_{\R^3\times \R^3}  \left( |v|^2+\frac{C}{|x-\xi_{0}|}\right)^{m/2}\,f_0(x,v)\, dx\, dv<+\infty.
\end{equation*}
Then there exists a global weak solution $(f,\xi)$ to the system
\eqref{syst:VP-bis} (here $C$ can be $C=0$, that is the standard Vlasov--Poisson system, or $C>0$), with $f\in C(\R_+,L^p(\R^3\times \R^3))\cap L^\infty(\R_+,L^\infty( \R^3\times \R^3))$ and $E\in C(\R_+,L^q(\R^3))$, for all  $2\leq q<3(3+m_0)/(6-m_0)$ if $m_0<6$, and $2\leq q\leq +\infty$ if $m_0>6$.

Moreover, for all $t\in \R_+$ and for all $m<\min(m_0,7)$,
%\begin{equation*}
%\iint_{\R^3\times \R^3}  \left( |v|^2+\frac{C}{|x-\xi_0|}\right)^{m/2}
%\,f(t,x,v)\, dx\, dv<+\infty.
%\end{equation*}
\begin{equation*}
\iint_{\R^3\times \R^3}  \left( |v|^2+\frac{1}{|x-\xi(t)|^{}}\right)^{m/2}
\,f(t,x,v)\, dx\, dv\leq C\,(1+t)^c,
\end{equation*}
where $C$ and $c$ only depend on $(f_0,\xi_0)$ and can be estimated explicitly.
\end{proposition}

\medskip

The  Cauchy problem for the Vlasov-Poisson system \eqref{syst:VP}, with or without point charge, has been the object of several works in the last decades. For the pure Vlasov-Poisson system without charge, namely $F=0$, global existence and uniqueness of classical solutions where obtained by Ukai and Okabe \cite{UO} in two dimensions. The three dimensional case is more delicate;  global weak solutions with finite energy were first built by Arsenev \cite{A}.

Global  existence and, in some cases, uniqueness, of more regular solutions were then separately established by Lions and Perthame \cite{LP} and by Pfaffelmoser \cite{Pf} by different techniques. In both works the main issue consists  in controlling the large plasma velocities for all time in order to propagate  regularity properties of the solution.

In \cite{LP}, this is achieved by constructing weak solutions with finite velocity moments of order higher than three
$$\iint_{\R^3\times \R^3} |v|^m f(t,x,v)\,dx\,dv<\infty,\quad m>3,$$
which, by Sobolev embeddings, implies further bounds on the spatial density and on the electric field. In particular, if the solution admits finite moments of order $m>6$ then the electric field is uniformly bounded and uniqueness holds under some additional regularity assumptions on the initial density. On the other hand, the theory of DiPerna and Lions \cite{DiLi} ensures that 
such solutions are constant along the trajectories of a "generalized flow" defined in a weak sense.

 In contrast with the eulerian approach of \cite{LP},  the strategy of \cite{Pf} relies on a careful analysis of the characteristics to control the growth of the velocity support and thereby obtain global existence and uniqueness of classical compactly supported solutions, which moreover propagate the regularity of the initial condition.

We refer to the  further improvements and developments by Schaeffer \cite{Sch}, Wollman \cite{W}, Castella \cite{castella}, Loeper \cite{L}, Chen and Zhang \cite{CZ}. Moreover, Gasser, Jabin and Perthame \cite{GJP} established propagation of the velocity moments for $m>2$ with an additional assumption on the space moments, and in \cite{Sal}, Salort proved existence and uniqueness of weak solutions even if $m<6$. Finally, Pallard \cite{Pallard} recently combined eulerian and lagrangian points of view to establish existence of solutions propagating  velocity moments for $m>2$.

\medskip

The study of the modified Vlasov-Poisson system with macroscopic point charges was initiated more recently by Caprino 
and  Marchioro \cite{CM}. In two dimensions, they proved global existence and uniqueness of solutions \`a la Pfaffelmoser. This  was then extended to the three-dimensional case by Marchioro, Miot and Pulvirenti \cite{MMP}. The results of \cite{CM} and \cite{MMP} hold for initial plasma densities that do not overlap the charge. Thanks to the repulsive nature of the plasma-charge interaction,  this property remains true at later times so that the field induced by the charge is bounded on the support of the density and the velocities of the plasma particles do not blow up.  The analysis of \cite{CM} and \cite{MMP} exploits the notion of microscopic energy, defined in this context by

$$h(t,x,v)=\frac{|v-\eta(t)|^2}{2}+\frac{1}{|x-\xi(t)|}.$$ It turns out that the variation of the energy along the plasma characteristics is controlled by the electric field (see Lemma \ref{lemma:pointwise-energy}), exactly as for the velocity in the absence of charge. On the other hand, the energy controls both the velocity and the distance to the charge. This makes it possible to adapt Pfaffelmoser's arguments by replacing the largest velocity of the plasma particles by the largest energy $\sup_{ \text{supp}(f(t))}h$, which by assumption is initially finite. Note that when the plasma density overlaps the charge this last quantity is not finite, so that the method fails. 

\medskip

In order to treat densities with unbounded energy, which is the purpose of the present paper, we adapt the PDE point of view from \cite{LP} and we show existence of a solution propagating the energy moments (see Definition \ref{mom} hereafter). In particular, since the energy moments control the velocity moments, we recover all additional regularity properties on the electric field which have been established in \cite{LP}.
 The new estimates that we provide are somehow based on an iteration of some of the arguments used in \cite{LP}, which leads to the improvement of both the assumptions made on the initial datum (that is the possibility of introducing a Dirac mass) and the dependence with respect to large time of the moments estimates (which become polynomial and explicit). Moreover, we also tried to provide fully detailed  proofs for the most technical part of \cite{LP}, for the sake of readability of our paper.  

We emphasize the fact that Theorem 
\ref{thm:main} allows for initial densities that do not necessarily vanish in a neighborhood of the charge. They nevertheless have to decay close to it somehow; unfortunately it does not include
the ``generic'' densities that are constant close to the charge.

On the other hand, our techniques do not  enable to obtain uniqueness because of the singularity of $F$ in the neighborhood of the charge.  Finally, we believe that the limitation $m_0<7$ (appearing in the proof of Proposition \ref{prop:global-estimate}) is purely technical. We also hope to extend Theorem \ref{thm:main} to the case of several point charges being all positively charged, as is the case in \cite{CM} and \cite{MMP}.

Among the most interesting open questions connected to this paper, we would like to mention: $i)$ the issues related to the gravitational case (which are the most interesting from the modeling point of view), where $|v|^2+\frac{1}{|x-\xi|}$ is replaced by $|v|^2-\frac{1}{|x-\xi|}$, which makes much more difficult the use of this quantity (see \cite{CMMP} for the two-dimensional case); $ii)$ the study of the Cauchy problem associated to system \eqref{syst:VP} with general initial densities around the charge, for instance constant, (note that this situation appears as a limit critical case of the present analysis); $iii)$ the improvement of large time estimates for moments, indeed in our work the polynomials appearing in the estimates have a very high degree (typically from 100 to 1000).

The remainder of this paper is organized as follows. Next Section is devoted to the proof of Theorem \ref{thm:main}. The general idea, which follows the lines of \cite{LP},
consists in deriving a priori estimates for the moments for a sequence
 of smooth solutions to \eqref{syst:VP}--\eqref{syst:VP-3} obtained by regularizing the initial density in order to obtain a global solution by compactness arguments. Concerning Proposition \ref{thm:main-bis} we omit the proof, since it can be obtained by mimicking the steps of the proof of Theorem \ref{thm:main}.

In Subsection \ref{subsec:general} we gather some basic facts and a priori estimates for the modified Vlasov-Poisson system \eqref{syst:VP}.  We also derive some first estimates for the energy moments. In Subsection \ref{subsec:duhamel} we introduce the flow associated to a relevant cutoff of the field $E+F$, which enables to express the solution of  \eqref{syst:VP} by means of Duhamel's formula with a suitable source term. Then, in Subsection \ref{subsection:small-time} we establish intermediate a priori estimates for the moments, which as a byproduct ensure that the moments are uniformly bounded for small times. These estimates are exploited to show that the moments are uniformly bounded for all times in Subsection \ref{subsection:large-time}.  They eventually provide a global solution satisfying the assumptions of Theorem \ref{thm:main}, as explained in Subsection \ref{sec:passing-limit}. An Appendix is also devoted to the proof of technical estimates on the flow.

\medskip

\section{Proof of Theorem \ref{thm:main}}

\textbf{Notations.}  Throughout the paper, $|z|$ denotes the usual (euclidian) norm of $z\in \R^3$ and $|A|$ is the operator norm of the matrix $A\in \mathcal{M}_3(\R)$. When $F: [0,T]\times \R^3\times \R^3\to  \R^3$ (or $\mathcal{M}_3(\R)$) and $1\leq p,q\leq +\infty$, we set
$\|F\|_{L^p([0,T],L^q(\R^3\times \R^3))}=\| \,|F|\, \|_{L^p([0,T],L^q(\R^3\times \R^3))}.$
%=\big\|\||F(t)|\|_{L^q(\R^3\times \R^3)}\big\|_{L^p([0,T]}.$$

\medskip

The notation $C$ will refer to a constant depending only on the quantities $\mathcal{H}(0)$, $\mathcal{M}(0)$, 
$\|f_0\|_{\infty}$, $|\xi_0|$, $|\eta_0|$, $m_0$, and $H_m(0)$, for $m<m_0$ (those
moments are defined in Subsection \ref{subsec:general},
cf. formula (\ref{ineq:energy-velo})), but not on $T$ (when a constant depends on another parameter, we state it explicitly).
Under the assumptions of Theorem \ref{thm:main}, all these quantities are finite (in the approximation  process, they will be bounded with respect to the regularization parameter).

\subsection{Interpolation estimates}

\medskip

We first recall a 
collection of well-known interpolation and Sobolev (or Young's) inequalities
that we shall apply later to the solutions of \eqref{syst:VP}--\eqref{syst:VP-3}.
All of them may be found in \cite{LP}.

\begin{proposition} \label{interp}
Let $f := f(x,v)\geq 0$. Let $b>a \ge 0$. Then for all $x\in\R^3$,
\begin{equation}
\label{estimate-moment}
\int_{\R^3}  |v|^{a}f(x,v)\,dv\leq C \|f\|_{L^{\infty}}^{\frac{b-a}{3+b}}\left( \int_{\R^3}
|v|^b f(x,v)\,dv\right)^{\frac{3+a}{3+b}},      
\end{equation}
with $C$ depending only on $a,b$.
In particular, setting $\rho(x)=\int_{\R^3} f(x,v)\,dv$, we have for any $b>0$,
\begin{equation} \label{int2}
 \|\rho\|_{L^{\frac{b+3}{3}}} \le C \|f\|_{L^{\infty}}^{\frac{b}{3+b}}\bigg(\intd |v|^b \,
f(x,v)\, dx\, dv \bigg)^{\frac3{3+b}}.
\end{equation}
\end{proposition}

\begin{proof}
For all $R\ge 0$
\begin{equation} \label{intR}
 \ints  |v|^{a}f(x,v)\,dv\le R^{a-b} \ints |v|^b \, f(x,v)\, dv + C R^{3+a} \|f\|_{\infty},
\end{equation}
and  estimate \eqref{estimate-moment} is obtained by optimizing $R>0$ (cf. also
the
proof of estimate (14) in \cite{LP}). Setting $a=0$ and taking the $L^{\frac{b+3}3}$ norm in
the $x$-variable, we obtain \eqref{int2}.

\end{proof}

\begin{proposition}\label{sobo} Let $f\geq 0$ be in $L^{1}(\R^3\times \R^3)$, 
assume that  $\rho(x) =\int_{\R^3} f(x,v)\,dv \in L^s(\R^3)$ (for some $1\leq s\leq +\infty$),
 and consider $E=\rho\ast (x\mapsto x/|x|^3)$. Then
for $1<s<3$, 
\begin{equation} \label{Sob}
 \| E\|_{L^{\frac{3s}{3-s}}} \le C\, \| \rho\|_{L^s},
\end{equation}
and for $s>3$,
\begin{equation} \label{Sobi}
 \| E\|_{L^{\infty}} \le C \| \rho\|_{L^s},
\end{equation}
where $C$ only depends on $s$.
\end{proposition}

\begin{proof}
The inequalities are direct consequences of Sobolev inequalities and the fact that $E = 4\pi \nabla_x \Delta_x^{-1} \rho$. 
%The last ones use
%\eqref{int2} and \eqref{Sobi}, setting $s=(m+3)/3$.
\end{proof}

\subsection{First dynamical estimates}
\label{subsec:general}

We now turn to the study of system \eqref{syst:VP}--\eqref{syst:VP-3}.  In the remainder of this article, we fix $T>0$. We will call \emph{classical solution} on $[0,T]$ any  solution  $(f,\xi)$  of \eqref{syst:VP}--\eqref{syst:VP-3} on $[0,T]$, with initial condition $(f_0,(\xi_0,\eta_0))$ satisfying the assumptions in Theorem \ref{thm:main}, such that \emph{moreover} $f_0$ is $C^1$, compactly supported, and vanishes in a neighborhood of $\xi_0$, which satisfies
$f\in C_c^1([0,T]\times \R^3\times \R^3)$ and $\xi\in C^2([0,T])$. Moreover,
$$f(t,\boldsymbol{x}(t,x,v),\boldsymbol{v}(t,x,v))=f_0(x,v),\quad  t\in \R_+,$$ where for all $t\in \R_+$, the map $(x,v)\in \R^3\setminus\{\xi_0\}\times \R^3\mapsto (\boldsymbol{x}(t,x,v),\boldsymbol{v}(t,x,v))\in \R^3\setminus\{\xi(t)\}\times \R^3$ is invertible and:

\noindent (i) For all $(x,v)\in \R^3\setminus\{\xi_0\}\times \R^3$,  $t\mapsto (\boldsymbol{x}(t,x,v),\boldsymbol{v}(t,x,v))\in C^1(\R_+)$ is the solution of
\begin{equation}
\label{syst:ode}
 \begin{cases}
\dsp \frac{d}{dt}{\boldsymbol{x}}(t,x,v)=\boldsymbol{v}(t,x,v)\\\vspace*{0.5em} 
\dsp \frac{d}{dt}{\boldsymbol{v}}(t,x,v)=E(t,\boldsymbol{x}(t,x,v))+\frac{\boldsymbol{x}(t,x,v)-\xi(t)}{|\boldsymbol{x}(t,x,v)-\xi(t)|^3},\quad (\boldsymbol{x},\boldsymbol{v})(0,x,v)=(x,v).
 \end{cases}
\end{equation}
In particular, we have for all $x\neq \xi_0$ and $v\in \R^3$,
 $$|\boldsymbol{x}(t,x,v)-\xi(t)|>0,\quad \forall t\in \R_+.$$
\noindent (ii) For all $t\in \R_+$, the map $(x,v)\mapsto (\boldsymbol{x}(t,x,v),\boldsymbol{v}(t,x,v))$ preserves the Lebesgue's measure on $\R^3\times \R^3$.

In the following we will sometimes write $(\boldsymbol{x}(t),\boldsymbol{v}(t))=(\boldsymbol{x}(t,x,v),\boldsymbol{v}(t,x,v))$.

 The existence (and uniqueness) of classical solutions corresponding to such initial data is ensured by \cite{MMP}.
Our 
purpose is to establish relevant a priori estimates for  $(f,\xi)$ on $[0,T]$, which will eventually lead
to the existence of a solution 
to \eqref{syst:VP}--\eqref{syst:VP-3} by compactness.
%  A la fin : citer les resultats de \cite{MMP} sur le fait qu'une telle solution existe
Such a priori estimates will concern the moments of order $ m<m_0$, which are defined in Definition~\ref{mom} below.

\medskip

 We start with a few basic properties of the Vlasov-Poisson system.
\begin{proposition}\label{eap_de_base}
Let $(f,\xi)$ be a 
%classical, compactlysupported
classical solution of \eqref{syst:VP}-\eqref{syst:VP-3} on $[0,T]$.

Then, the norms
\begin{equation*}\|f(t)\|_{L^p({\R^3\times \R^3}) },\quad 1\leq p\leq \infty,
\end{equation*}
and the energy
\begin{equation*}
\begin{split}
\mathcal{H}(t) & =\frac12\iint_{\R^3\times \R^3}  |v|^2 f (t,x,v)\, dv\, dx + \frac12 | \eta(t)|^2\\
&+ \frac12 \iint_{\R^3\times \R^3}  \frac{\rho(t,x)\, \rho(t,y)}{|x-y|} \, dx dy   + 
\int_{\R^3}  \frac{\rho(t,x)}{|x - \xi(t)|} \, dx    
\end{split}
\end{equation*}
are conserved in time. 
In particular, the mass
\begin{equation*}
\mathcal{M} (t) = \iint_{\R^3\times \R^3}   f (t,x,v)\, dx\, dv
\end{equation*}
is conserved in time.
\end{proposition}

\begin{proof}
The conservation of the $L^{p}$ norms is an immediate consequence of the fact that $f$ is constant along the trajectories of a Lebesgue's measure preserving flow.

We only detail the computation for the energy:
\begin{equation*}\begin{split}
&\frac{d}{dt} \bigg\{ \iint_{\R^3\times \R^3}  f\, \frac{|v|^2}2 \, dx\, dv + \frac12 |\eta|^2
+ \frac12 \iint_{\R^3\times \R^3}  \frac{\rho(x)\, \rho(y)}{|x-y|} \, dx \,dy
 + \int_{\R^3} \frac{\rho(x)}{|x - \xi|} \, dx\bigg\}\\ 
&= \iint_{\R^3\times \R^3}  v\cdot (E + F)\, f\, dx\,dv +  \eta \cdot 
E(\xi)   \\
& - \iint_{\R^3\times \R^3}  \frac{\nabla_x \cdot (\ints vf\,dv)}{|x-y|} \, \rho(y)\, dx\,dy
- \int_{\R^3} \frac{ \nabla_x \cdot \int v\,f\,dv}{|x- \xi|} \,
dx  - \int_{\R^3} \rho(x) \, \eta \cdot \frac{\xi - x}{|\xi
- x|^3} \, dx\\
& = 0.\end{split} \end{equation*}
\end{proof}

For the initial data $(f_0,(\xi_0,\eta_0))$  considered in the setting of Theorem~\ref{thm:main}, the energy is initially finite; indeed Proposition \ref{interp} yields $\rho_0\in L^1\cap L^{5/3}$, so that $\iint \rho(x)\rho(y)/|x-y|\,dx\,dy$ is finite by H\"older estimates; on the other hand the other terms are clearly finite by assumption $(ii)$. So we immediately get the
\begin{proposition}\label{prop:dist}
Let $(f,\xi)$ be a classical solution of \eqref{syst:VP}-\eqref{syst:VP-3} on $[0,T]$.
We have
\begin{equation}\label{prop:velocities}
\sup_{t\in [0,T]}|\eta(t)|\leq \sqrt{2\mathcal{H}(0)}
\end{equation}
and
\begin{equation}\label{prop:dist2}
\sup_{t\in [0,T]}|\xi(t)|\leq |\xi_0|+\sqrt{2\mathcal{H}(0)}\,T.
\end{equation}
\end{proposition}

\begin{proof}
The first inequality is a consequence of the conservation
of the energy. The second one comes out of the integration w.r.t. time of
the first one.
\end{proof}

Another well-known consequence of the conservation of the energy is the following
\begin{proposition}\label{prop:norm-star}
Let $(f,\xi)$ be a classical solution of \eqref{syst:VP}-\eqref{syst:VP-3} on $[0,T]$.
We have
\begin{equation*}
\sup_{t\in [0,T]}\|\rho(t)\|_{L^{5/3}}\leq C,
\end{equation*}
and for all $\frac{3}{2}<r\leq \frac{15}{4}$,
\begin{equation*}
\sup_{t\in [0,T]}\|E(t)\|_{L^r}\leq C,
\end{equation*}
with $C$ a constant depending only on $\mathcal{H}(0)$, $\|f_0\|_{\infty}$ and $r$.

\end{proposition}

\begin{proof}
The first estimate is a consequence of \eqref{int2} with $b=2$ and the fact that the velocity moment of order $2$ is controlled by the energy. The second estimate is deduced from the first one and \eqref{Sob}, using the fact that $\rho\in L^\infty([0,T], L^1\cap L^{5/3}(\R^3))$.

\end{proof}

We now give our definition of energy moments.
\begin{definition}\label{mom}
We define the energy function
\begin{equation*}
h(t,x,v) =\frac{|v-\eta(t)|^2}{2}+\frac{1}{|x-\xi(t)|}+\mathcal{H}(0)+(\mathcal{M}_0)^{-1}+1.
\end{equation*}
 In view of Proposition~\ref{prop:dist}, we have 
\begin{equation*}
|v|\leq 2\sqrt{h}(t,x,v)\quad \forall (t,x,v)\in[0,T]\times \R^3\times \R^3.
\end{equation*}

Then we set for $k\in \R_+$
\begin{equation}\label{energy}
\tilde{H}_{k}(t) =\iint_{\R^3\times \R^3}  h(t,x,v)^{k/2} f(t,x,v) dx\, dv,
\end{equation}
and
\begin{equation} \label{deffina}
H_k(t) =\sup_{s\in[0,t]}\tilde{H}_{k}(s)=\sup_{s\in[0,t]} \iint_{\R^3\times \R^3}  
h(s,x,v)^{k/2} f(s,x,v) dx\, dv.
\end{equation}
In particular 
\begin{equation*} 
H_k(t) \geq 1.
\end{equation*}

\end{definition}

A first basic observation is that the energy moments $H_k$ control the velocity moments $M_k$ defined in \cite{LP}, namely
\begin{equation}\label{ineq:energy-velo}
M_k(t) := \sup_{s\in[0,t]} \iint_{\R^3\times \R^3}  |v|^k \, f(s,x,v) \, dx\,dv\leq 2^k H_k(t).
\end{equation}

\medskip

\begin{lemma}\label{lemma:pointwise-energy}
Let $(f,\xi)$ be a classical solution of \eqref{syst:VP}-\eqref{syst:VP-3} on $[0,T]$. For all plasma trajectory  $(\bold{x}(t),\bold{v}(t))=(\bold{x}(t,x,v),\bold{v}(t,x,v))$ solution of \eqref{syst:ode}, we have
\begin{equation*}
\frac{d}{dt}\sqrt{h}(t,\bold{x}(t),\bold{v}(t))\leq |E(t,\bold{x}(t))|+|E(t,\xi(t))|.
\end{equation*}
\end{lemma}

\begin{proof}
A simple computation using \eqref{syst:ode} and \eqref{syst:VP-2}  yields
$$\frac{d}{dt}{h}(t,\bold{x}(t),\bold{v}(t))=(\bold{v}(t)-\eta(t))\cdot \big( E(t,\bold{x}(t))- E(t,\xi(t))\big).
$$ Note that the singular field $F$ does not appear in the equality.
\end{proof}

\begin{lemma} 
Let $(f,\xi)$ be a classical solution of \eqref{syst:VP}-\eqref{syst:VP-3} on $[0,T]$.
We have for all $t\in [0,T]$ and for all $k\in \mathbb{R}_+$
 \begin{equation} \label{evol_mom_diff}
\frac{d}{dt}\tilde{H}_{k}(t)\leq C(k)\,
\Big( \| E(t)\|_{L^{k+3}}+\left|E(t,\xi(t))\right|\Big)
H_k(t)^{\frac{k+2}{k+3}}
\end{equation}
and therefore,
\begin{equation} \label{evol_mom_integ}
H_{k}(t)\leq C(k)
\left\{ H_{k}(0)+\left( \int_0^t \Big\{\| E(s)\|_{L^{k+3}}+
\left|E(s,\xi(s))\right|\Big\}\,ds\right)^{k+3}\right\}.
\end{equation}
\label{lemma:est-moments}
\end{lemma}

\begin{proof}
Since $f\in C_c^1([0,T]\times \R^3\times\R^3)$ is a classical solution of \eqref{syst:VP}, we may compute
\begin{equation*}\begin{split}
\frac{d}{dt} \tilde{H}_{k}(t)&=\frac{k}{2}\,
\iint_{\R^3\times \R^3}  h^{k/2-1}f\, \big\{\partial_t h+v\cdot \nabla_x h+(E+F )\cdot \nabla_v
h\big\}(t,x,v)\,dx\,dv\\
&=\frac{k}{2}
\iint_{\R^3\times \R^3}   h^{k/2-1} f\,\big\{(v-\eta(t))\cdot \big( E(t,x)- E(t,\xi(t))\big)\big\}\,dx\,dv.
\end{split}\end{equation*}
We remark again that the choice of the energy function $h$ enabled to get rid of the singular field in the second equality.

Therefore
\begin{equation}
\label{ineq:est-mom-1}
\begin{split}
\frac{d}{dt} \tilde{H}_{k}(t)&\leq C(k)
\iint_{\R^3\times \R^3}  | E(t,x)|h^{(k-1)/2}f (t,x,v)\,dx\,dv\\
&+ C(k)\, |E(t,\xi)|\, \iint_{\R^3\times \R^3} 
h^{(k-1)/2}f(t,x,v)\,dx\,dv.
\end{split}\end{equation}

We assume first that $k\geq 1$.
In order to bound the first term of the right-hand side in \eqref{ineq:est-mom-1}, we use interpolation arguments from \cite{LP} that we recall here for sake of clarity.
First, we have thanks to H\"older's inequality
\begin{equation*}
\iint_{\R^3\times \R^3} | E(t,x)|h^{(k-1)/2}f(t,x,v)\,dx\,dv\leq \|E(t)\|_{L^{k+3}} \left \|  \int_{\R^3} h^{(k-1)/2}
 f (t,\cdot,v) \,dv   \right \|_{L^{\frac{k+3}{k+2}}}.
\end{equation*}
Next, we have for $x\in \R^3$ and for $R>0$,
\begin{equation*}
\begin{split}
 \int_{\R^3} h^{(k-1)/2}   f(t,x,v)  \,dv &= \int_{h^{1/2}\leq R}  h^{(k-1)/2}   f (t,x,v) \,dv
+ \int_{h^{1/2}\geq R} h^{(k-1)/2}   f(t,x,v)  \,dv \\
 &\leq R^{k-1}\int_{|v|\leq CR}f(t,x,v)\,dv+R^{-1} \int_{h^{1/2}\geq R} h^{k/2}
f (t,x,v) \,dv\\
 &\leq C \|f(t)\|_{L^{\infty}}R^{k+2}+R^{-1} \int_{\R^3}  h^{k/2}   f(t,x,v)  \,dv.
\end{split}
\end{equation*}
We used the fact  that $|v|\leq Ch^{1/2}$ in the
second inequality. Now, optimizing w.r.t. $R$, and using the identity
$\|f(t)\|_{L^{\infty}}=\|f_0\|_{L^{\infty}}$, we find
\begin{equation*}
 \int_{\R^3} h^{(k-1)/2}   f(t,x,v)  \,dv
 \leq C(k)\, \left( \int_{\R^3} h^{k/2}   f (t,x,v) \,dv\right)^{(k+2)/(k+3)}.
\end{equation*}
So finally, integrating in $x$, we obtain
\begin{equation*}
 \left\|\int_{\R^3} h^{(k-1)/2}   f (t,\cdot,v) \,dv\right\|_{L^{\frac{k+3}{k+2}}}
 \leq C(k)\,  H_{k}(t)^{(k+2)/(k+3)},
\end{equation*}
and we are led to
\begin{equation}\label{ineq:interpol1}
\iint_{\R^3\times \R^3}  | E(t,x)|h^{(k-1)/2}f(t,x,v)\,dx\,dv\leq C(k)\, \|E(t)\|_{L^{k+3}}\,H_{k}(t)^{(k+2)/(k+3)}.
\end{equation}

\medskip

We next estimate the second term in \eqref{ineq:est-mom-1}. Applying again H\"older's inequality yields
\begin{equation*}
 \iint_{\R^3\times \R^3}  h^{(k-1)/2} \,  f(t,x,v)  \,dx\,dv
 \leq  \left(   \iint_{\R^3\times \R^3}     f(t,x,v)  \,dx\,dv  \right) ^{1/k} \left(   \iint_{\R^3\times \R^3}   h^{k/2}   \,f(t,x,v)  \,dx\,dv  \right) ^{(k-1)/k}
\end{equation*}
so that, since $\mathcal{M}(t)=\mathcal{M}_0$,
\begin{equation*}
\iint_{\R^3\times \R^3}  h^{(k-1)/2}f(t,x,v)\,dx\,dv\leq C(k)\, H_{k}(t)^{(k-1)/k}  .
\end{equation*}
Since $(k-1)/k\leq (k+2)/(k+3)$ and $H_k(t)\geq 1$, it follows that
\begin{equation}\label{ineq:interpol2}
|E(t,\xi)|\iint_{\R^3\times \R^3}  h^{(k-1)/2}f(t,x,v)\,dx\,dv\leq C(k)\, |E(t,\xi)|  H_{k}(t)^{(k+2)/(k+3)}  .
\end{equation}

Gathering estimates \eqref{ineq:interpol1} and \eqref{ineq:interpol2} we
are led to the conclusion
of Lemma \ref{lemma:est-moments}.

\medskip

When $k<1$, since $h(t,x,v)\geq 1$ we directly obtain 
\begin{equation*}
\begin{split}
\frac{d}{dt} \tilde{H}_{k}(t)&\leq C(k)\left(\|E(t)\|_{L^{k+3}} \left \|  \rho(t)\  \right \|_{L^{\frac{k+3}{k+2}}}
+ \mathcal{M}_0 |E(t,\xi)|\right)\\
&\leq C(k)\Big(\|E(t)\|_{L^{k+3}}+  |E(t,\xi)|\Big)\leq C(k)\Big( \| E(t)\|_{L^{k+3}}+\left|E(t,\xi)\right|\Big)
H_k(t)^{\frac{k+2}{k+3}},
\end{split}\end{equation*}
because (since $1\leq (k+2)/(k+3)\leq 5/3$)
we know that $\|\rho(t)\|_{L^{\frac{k+2}{k+3}}}\leq C$ by Proposition \ref{prop:norm-star},
 and $H_k(t)\geq 1$.

\end{proof}

\bigskip

In order to exploit  Lemma \ref{lemma:est-moments} to estimate the moments $H_k(t)$ by a Gronwall inequality, we now need to control the electric fields
$\|E(t)\|_{L^{k+3}}$ and  $ |E(t,\xi)|$.  The next Subsection will be devoted to the control of $\|E(t)\|_{L^{k+3}}$. On the other hand,  when $\mathcal{M}_0$ is not too large,
 one can control the quantity $|E(t,\xi)|$ by the following virial-type argument.

 \begin{proposition}
 \label{prop:virial}
Let $(f,\xi)$ be a classical solution of \eqref{syst:VP}-\eqref{syst:VP-3} on $[0,T]$.
We have
 \begin{equation*}
 \int_0^t \left|E(s,\xi(s))\right|\,ds\leq C (1+t),
 \end{equation*}
and 
\begin{equation*}
 \int_0^t \iint_{\R^3\times \R^3} \frac{f(s,x,v)}{|x-\xi(s)|^2}\,dx\,dv\,ds\leq C(1+t).
\end{equation*}
 \end{proposition}

\begin{remark} This is the only point in the proof of Theorem \ref{thm:main} where we have to
use assumption $(i)$.\end{remark}

\begin{proof}Let $(\boldsymbol{x}(s),\boldsymbol{v}(s))=(\boldsymbol{x}(s,x,v),\boldsymbol{v}(s,x,v))$ be a plasma trajectory on $[0,T]$. Using the system of ODE  \eqref{syst:VP-2} and  \eqref{syst:ode}, we compute 
\begin{equation*}\begin{split}
\frac{d^2}{ds^2}|\boldsymbol{x}(s)-\xi(s)|&=\frac{|\boldsymbol{v}(s)-\eta(s)|^2}{|\boldsymbol{v}(s)-\xi(s)|}+\frac{1}{|\boldsymbol{x}(s)-\xi(s)|^2}\\&
+\frac{\big(\boldsymbol{x}(s)-\xi(s)\big)\cdot \big(E(s,\boldsymbol{x}(s))-E(s,\xi(s))\big)}{|\boldsymbol{x}(s)-\xi(s)|}
-
\frac{[(\boldsymbol{x}(s)-\xi(s))\cdot (\boldsymbol{v}(s)-\eta(s))]^2}{|\boldsymbol{x}(s)-\xi(s)|^3}.\end{split}
\end{equation*}

Therefore
\begin{equation}\label{ineq:virial}
\frac{1}{|\boldsymbol{x}(s)-\xi(s)|^2}\leq \frac{d^2}{ds^2}|\boldsymbol{x}(s)-\xi(s)|+|E(s,\boldsymbol{x}(s))|+|E(s,\xi(s))|.
\end{equation}

On the other hand, since $f$ is constant along the trajectories of the measure-preserving flow $(\boldsymbol{x},\boldsymbol{v})$, we have by changing variable
\begin{equation*}\begin{split}
\left|E(s,\xi(s))\right|&\leq \iint_{\R^3\times \R^3}  \frac{f(s,x,v)}{|x-\xi(s)|^2}\,dx\,dv= \iint_{\R^3\times \R^3}  \frac{f_0(x,v)}{|\boldsymbol{x}(s,x,v)-\xi(s)|^2}\,dx\,dv.
\end{split}\end{equation*}
Therefore inserting \eqref{ineq:virial} we get
\begin{equation}\label{ineq:virial-2}
\begin{split}
\int_0^t \left|E(s,\xi(s))\right|\,ds
&\leq \iint_{\R^3\times \R^3}  f_0(x,v)\left(\int_0^t\frac{d^2}{ds^2}|\boldsymbol{x}(s)-\xi(s)|\,ds\right)\,dx\,dv\\&+\int_0^t \left(\iint_{\R^3\times \R^3}  f_0(x,v)|E(s,\boldsymbol{x}(s,x,v))|\,dx\,dv\right)\,ds+\mathcal{M}_0\int_0^t |E(s,\xi(s))|\,ds.
\end{split}
\end{equation}

For the first term in the right-hand side of \eqref{ineq:virial-2}, we have
\begin{equation*}\begin{split}
\iint_{\R^3\times \R^3}  f_0(x,v)\left(\int_0^t\frac{d^2}{ds^2}|\boldsymbol{x}(s)-\xi(s)|\,ds\right)\,dx\,dv&
=\iint_{\R^3\times \R^3}   f_0(x,v)\left[\frac{d}{ds}|\boldsymbol{x}-\xi|\right]_{s=0}^{s=t}\,dx\,dv\\
&\leq \iint_{\R^3\times \R^3}   f_0(x,v)\left(\Big|\frac{d}{ds}|\boldsymbol{x}-\xi|\Big|(t)+\Big|\frac{d}{ds}|\boldsymbol{x}-\xi|\Big|(0)\right)\,dx\,dv\\
&\leq 2\sup_{t\in[0,T]} \iint_{\R^3\times \R^3}  f_0(x,v)|\boldsymbol{v}(t,x,v)-\eta(t)|\,dx\,dv\\
&= 2\sup_{t\in[0,T]}\iint_{\R^3\times \R^3}  f(t,x,v)|v-\eta(t)|\,dx\,dv.
\end{split}
\end{equation*}
Hence, by H\"older's inequality, we obtain
\begin{equation}\label{ineq:asugar}\begin{split}
\iint_{\R^3\times \R^3}  f_0(x,v)\left(\int_0^t\frac{d^2}{ds^2}
|\boldsymbol{x}(s)-\xi(s)|\,ds\right)\,dx\,dv&\leq C
 \sup_{t\in[0,T]} \mathcal{M}(t)^{1/2}\,(\mathcal{H}(t) + \mathcal{H}(t)\,  \mathcal{M}(t))^{1/2}\leq C.
\end{split}\end{equation}

We turn to the second term in \eqref{ineq:virial-2}. We have
 by changing variable backwards
\begin{equation}\label{ineq:pao}
\begin{split}
\int_0^t \left(\iint_{\R^3\times \R^3}  f_0(x,v)|E(s,\boldsymbol{x}(s,x,v))|\,dx\,dv\right)\,ds&=\int_0^t \left(\iint_{\R^3\times \R^3}  f(s,x,v)|E(s,x)|\,dx\,dv\right)\,ds\\
&=\int_0^t \left(\int_{\R^3} \rho(s,x)|E(s,x)|\,dx\right)\,ds\\
&\leq C\int_0^t \|\rho(s)\|_{L^{5/3}}\|E(s)\|_{L^{5/2}}\,ds\leq C t.
\end{split}
\end{equation}
We used Proposition \ref{prop:norm-star} in the last inequality.

Therefore coming back to \eqref{ineq:virial-2}, we find
\begin{equation*}
\begin{split}
\int_0^t \left|E(s,\xi(s))\right|\,ds
&\leq C (1+t)+\mathcal{M}_0\int_0^t \left|E(s,\xi(s))\right|\,ds.
\end{split}
\end{equation*}
The first inequality of Proposition \ref{prop:virial} then follows from the assumption $(i)$ on $\mathcal{M}_0$ since $\lambda\leq 1$.
For the second inequality, we come back to \eqref{ineq:virial}, integrate with respect to the measure $f(0,x,v)\,dx\,dv\,ds$, and use \eqref{ineq:asugar} and \eqref{ineq:pao}. The conclusion follows.
\end{proof}

\subsection{A Duhamel formula}

\label{subsec:duhamel}

The purpose of this Subsection is to establish estimates for the norms $\|E(t)\|_{k+3}$.

\medskip

Let $\chi_0$ be a smooth cutoff function such that $\chi_0(x)=1$ on $B(0,1)$, $\chi_0(x)=0$ on $B(0,2)^c$ 
and $0<\chi_0(x)<1$ on $\R^3$, $\| \nabla\chi\|_{L^{\infty}} \le 2$, $\| \nabla\nabla\chi\|_{L^{\infty}} \le 20$. 
Let $R:=R(T)>1$ be sufficiently large (depending only on  $\|f_0\|_{L^1}$
 and $T$). Such an $R:=R(T)$ will be determined by the condition  \eqref{cond:R} in the Appendix, by
 $$ R(T)= 24^{1/3}\,K_0^{1/3} (1 + \|f_0\|_{L^1})^{1/3}\, (1+T), $$
with $K_0$ a given number such that $K_0\geq 100$.

We set  $\chi_R(z)=\chi_0(z/R)$. 

We decompose the electric field and the force field into two parts:
\begin{equation*}
\begin{split}
E=E_{\text{int}}+E_{\text{ext}}, \qquad F = F_{\text{int}}+F_{\text{ext}},
\end{split}
\end{equation*}
where 
\begin{equation*}%\label{def:champ-int}
E_{\text{int}}=\rho\ast ( x\mapsto \chi_R(x)\, x/|x|^3), \qquad
F_{\text{int}}(t,x) = F(t,x)\,\chi_R(x-\xi(t)).
\end{equation*}
 
We have (denoting by superscripts the components of the fields)
%for a constant $C(\chi_0)>1$ depending only on $\|\chi_0\|_{W^{2,\infty}}$:
\begin{equation}\label{estimate:Eext-0}
 \max_{i=1,2,3}\|E^{(i)}_{\text{ext}}+F^{(i)}_{\text{ext}}\|_{L^\infty}\leq \frac{1+\|f_0\|_{L^1}}{R^2},
\end{equation}
\begin{equation}\label{estimate:Eext-1}
\max_{i,j=1,2,3}\|\pa_{x_j} (E^{(i)}_{\text{ext}}+F^{(i)}_{\text{ext}})\|_{L^\infty}\leq 6\,\frac{1+\|f_0\|_{L^1}}{R^3},
\end{equation}
and
\begin{equation}\label{estimate:Eext-2}
\max_{i,j,k=1,2,3} \|\pa_{x_jx_k}^2 (E^{(i)}_{\text{ext}}+F^{(i)}_{\text{ext}})\|_{L^\infty}
\leq  60\,\frac{1+\|f_0\|_{L^1}}{R^4}.
\end{equation}

\medskip

As in \cite{LP}, we  write the Vlasov-Poisson equation using the
internal part of $E$ and $F$ as a source term: 
\begin{equation}\label{source-term}
\partial_t f+v\.\nabla_x f+(E_{\text{ext}}+F_{\text{ext}})\.\nabla_v f=-(E_{\text{int}}+F_{\text{int}})\.\nabla_v f.
\end{equation}
The reason why we do not consider the free transport (namely we do not consider the full field as a source term) will appear in Subsection \ref{subsection:large-time}.

\medskip

In all this Subsection, we fix $t\in[0,T]$. We define the flow map $(x,v)\mapsto (X^t,V^t)(x,v)$ such that for all 
$s\in[0,t]$,
\begin{equation}\label{characteristics-2}
\left\{
\begin{array}{ll}
\dsp \frac{d}{ds}X^t(s,x,v) =-V^t(s,x,v),&X^t(0,x,v)=x, \\
\dsp \frac{d}{ds}V^t(s,x,v)=-(E_{\text{ext}}+F_{\text{ext}})(t-s,X^t(s,x,v)),&V^t(0,x,v)=v.
\end{array}
\right.
\end{equation}
This is the backward flow associated to the field $E_{\text{ext}}+F_{\text{ext}}$. More precisely, if $(X_f,V_f):\R_+\times \R^3\times \R^3\to \R^3$ is the forward flow of $E_{\text{ext}}+F_{\text{ext}}$, defined by $\frac{d}{dt}{X}_f(t,x,v)=V_f(t,x,v), \frac{d}{dt}{V}_f(t,x,v)= (E_{\text{ext}}+F_{\text{ext}})(t, X_f(t,x,v))$ and $(X_f,V_f)(0,x,v)=(x,v)$,
 then for all $t\in \R_+$, we have $(X_f,V_f)(t,\cdot,\cdot)^{-1}=(X^t,V^t)(t,\cdot,\cdot)$.

\medskip

In the sequel we shall omit the dependence upon $t$ in the notations. Moreover,
  we shall sometimes write $$(X(s),V(s)) =(X^t(s,x,v),V^t(s,x,v))=(X(s,x,v),V(s,x,v)).$$
 Then, $(X(s),V(s))$ preserves the Lebesgue's measure on $\R^3\times \R^3$. Note that if the external field vanished we would obtain the free flow $X(s,x,v)=x-vs,V(s,x,v)=v$, and if we considered the total field $E+F$ in \eqref{characteristics-2}, we would obtain 
$(X,V)(t)=(\boldsymbol{x},\boldsymbol{v})(t)^{-1}$.

Using the invertibility properties of the flow listed in the Appendix,
 one can establish the analogue of Proposition \ref{interp}:
\begin{proposition}
\label{interp-modified}Let $(f,\xi)$ be a classical solution of \eqref{syst:VP}-\eqref{syst:VP-3} on $[0,T]$. Let $t\in[0,T]$ and let $(\tau,s)\in[0,t]^2$. Let $b>a \ge 0$. Then
\begin{equation}
\label{estimate-moment-modified}
\int_{\R^3}  |v|^{a}f(\tau,X(s,x,v),V(s,x,v))\,dv\leq C(a,b)\, 
\|f_0\|_{L^{\infty}}^{\frac{b-a}{3+b}}\left( \int_{\R^3}
|v|^b f(\tau,X(s,x,v),V(s,x,v))\,dv\right)^{\frac{3+a}{3+b}}      .
\end{equation}

Moreover, setting  $\tilde{\rho}(\tau,s,x)=\int_{\R^3} f(\tau,X(s,x,v),V(s,x,v))\,dv$, we have
\begin{equation} \label{int2-modified}
\|\tilde{\rho}(\tau,s,\cdot)\|_{L^{ \frac{b+3}{3}}}
\leq C(b)\,(1+s)^{\frac{3b}{3+b}} \|f_0\|_{L^{\infty}}^{\frac{b}{3+b}}\bigg(1+\iint_{\R^3\times \R^3}  |v|^b \,
f(\tau,x,v)\, dx\, dv \bigg)^{\frac3{3+b}} .
\end{equation}
\end{proposition}

\begin{proof}
 For the first inequality this is exactly the same proof as for Proposition
\ref{interp}, estimate \eqref{estimate-moment}. The second inequality is obtained thanks to the bound $|V(s,x,v)-v|\leq s$ (see \eqref{nnntrois}), and using the fact that $(X(s),V(s))$ preserves the Lebesgue's measure. Recall also that $\|f(\tau)\|_{L^\infty}=\|f_0\|_{L^{\infty}}$.
\end{proof}

A consequence of Propositions \ref{sobo} and \ref{interp-modified} is
\begin{proposition}\label{sobo-modified} Let $(f,\xi)$ be a classical solution of \eqref{syst:VP}-\eqref{syst:VP-3} on $[0,T]$. Let $t\in[0,T]$ and let $(\tau,s)\in[0,t]^2$.
 Let  $\tilde{\rho}(\tau,s,x) =\int_{\R^3} f(\tau,X(s,x,v),V(s,x,v))\,dv$ and $\tilde{E}(\tau,s, \cdot)=\tilde{\rho}(\tau,s, \cdot)\ast ( x \mapsto x/|x|^3)$. If $0<m <6$, we have
\begin{equation} \label{ineq:norm-field-modified}
 \| \tilde{E}(\tau,s)\|_{L^{\frac{3(m+3)}{6-m}} } \leq C(1+s)^{\frac{3m}{3+m}}\left(1+\iint_{\R^3\times\R^3} |v|^mf(\tau,x,v)\,dx\,dv\right)^{\frac{3}{m+3}}\leq C\, (1+s)^{\frac{3m}{3+m}}H_m(\tau)^{\frac{3}{m+3}},
\end{equation}
and if $m>6$,
\begin{equation} \label{ineq:bound-field-modified}
 \| \tilde{E}(\tau,s)\|_{L^{\infty}} \le C\, (1+s)^{\frac{3m}{3+m}}\left(1+\iint_{\R^3\times\R^3} |v|^mf(\tau,x,v)\,dx\,dv\right)^{\frac{3}{m+3}}\leq C  (1+s)^{\frac{3m}{3+m}} H_m(\tau)^{\frac{3}{m+3}},
\end{equation}
where $C>0$ depends only on $\|f_0\|_{L^{\infty}}$ and $m$.
\end{proposition}

\begin{proof}
We apply \eqref{Sob} and \eqref{Sobi} with $s=(m+3)/3$, and conclude thanks to \eqref{int2-modified} 
with $b=m$.
\end{proof}

\medskip

Using Duhamel formula we can express the solution of \eqref{source-term} as follows:
\begin{equation}\begin{split}\label{eq:duhamel}
f(t,x,v)=f_0(X(t,x,v),V(t,x,v))-\int_{0}^t\big( \dev_v [(E_{\text{int}}+F_{\text{int}})\,f]\big)(t-s,X(s,x,v),V(s,x,v))\, ds.
\end{split}
\end{equation}

\begin{proposition}\label{prop:duhamel-0} 
Let $(f,\xi)$ be a classical solution of \eqref{syst:VP}-\eqref{syst:VP-3} on $[0,T]$.
We have for  $3\leq m <m_0$,  
\begin{equation} \label{eqE}
\|E(t)\|_{L^{m+3}}\leq C\, (1+T)^{\sup(2, \frac{3m}{3+m})} +C\, \left\|\int_0 ^t s
\int_{\R^3}(|E_{\text{int}}+F_{\text{int}}|f)(t-s,X(s),V(s))\,dv \, ds\right\|_{L^{m+3}},
\end{equation}
where $C$ depends on $m$ and the initial data.
\end{proposition}

\begin{proof}
By \eqref{eq:duhamel}, we have
\begin{equation} \label{Srho}\begin{split}
 \rho(t,x) &= \int_{\R^3} f_0(X(t,x,v), V(t,x,v))\, dv \\-  \int_0^t& \int_{\R^3} \big(\dev_v
((E_{\text{int}}+F_{\text{int}})\,f)\big)(t-s,X(s,x,v),V(s,x,v))\, dv\,ds\\
& :=\rho_1(t,x)+\rho_2(t,x).\end{split}
\end{equation}

We set $E_1(t)=\rho_1(t)\ast (x\mapsto x/|x|^3)$. We apply Propositions \ref{interp-modified} and \ref{sobo-modified}  with $\tau=0$ and $s=t$: first, since
$H_m(0)$ is finite, we have (with $\frac{3(m+3)}{6-m}$ replaced by $\infty$ if $m>6$) 
\begin{equation*}
\sup_{t\in[0,T]}\left\|E_1(t)\right\|_{L^{\frac{3(m+3)}{6-m}}}\leq C\, (1+T)^{\frac{3m}{3+m}}.
\end{equation*}
On the other hand, by \eqref{int2-modified}, we have $\|\rho_1(t)\|_{L^{5/3}}\leq C\, (1+T)^{6/5}\, (1+H_2(0)^{3/5})
\leq C\, (1+T)^{6/5}$. Therefore by interpolation, observing that $5/3\leq m+3\leq 3(m+3)/(6-m)$,
 we get
\begin{equation}
\label{eq:E1estiamte}
\sup_{t\in [0,T]}\left\|E_1(t)\right\|_{{m+3}}\leq C(1+T)^{\frac{3m}{3+m}}.
\end{equation}

For the term $\rho_2$ and the corresponding field $E_2$,
 we have to work more. More precisely, we will establish in the Appendix, see \eqref{nrho}, that
\begin{equation}\label{eq:rho2}
\begin{split}
&\rho_2(t,x) =\dev_x  \int_0^t\int_{\R^3} N(s,x,v) [(E_{\text{int}}+F_{\text{int}})f](t-s,X(s,x,v),V(s,x,v))\,dv\,ds\\
&+ \int_0^t\int_{\R^3} ( \dev_v (^tM)(x,v)-\dev_x (^tN)(x,v))
[ (E_{\text{int}}+F_{\text{int}})f](t-s,X(s,x,v),V(s,x,v))\,dv\,ds\\
& :=\rho_{2,1}(t,x)+\rho_{2,2}(t,x)\\
& :=\dev_x a_{2,1}(t,x)+\rho_{2,2}(t,x),
\end{split}
\end{equation}
where
\begin{equation*}\begin{split}
M&=M(s,x,v)=[D_v V- D_xV(D_x X)^{-1}D_v X]^{-1}, \\ 
N&=N(s,x,v)=(D_x X)^{-1}(D_v X) M(s,x,v).
\end{split}
\end{equation*}

We first estimate $E_{2,1}(t)=\rho_{2,1}(t)\ast (x \mapsto x/|x|^3)$. By elliptic regularity we have
\begin{equation*}
\|E_{2,1}(t)\|_{L^{m+3}}\leq C\|a_{2,1}(t)\|_{L^{m+3}}.
\end{equation*}
Next, by \eqref{nnnun} we have $\|N(s)\|_{L^{\infty}(\R^3\times\R^3)}\leq C s$, so that
\begin{equation}\label{eq:rho12}
\|E_{2,1}(t)\|_{L^{m+3}}\leq C\,
\left\|\int_0^t s \int_{\R^3} (|E_{\text{int}} + F_{\text{int}}|\,f)(t-s, X(s,x,v), V(s,x,v)) \, dv\, ds\right\|_{L^{m+3}}.\end{equation}

We next estimate $E_{2,2}(t)=\rho_{2,2}(t)\ast  (x \mapsto  x/|x|^3) $. By interpolation inequality \eqref{Sob} with $s=3(m+3)/(m+6)\in ]1,3[$,
 we have $\|E_{2,2}(t)\|_{L^{m+3}}\leq C\|\rho_{2,2}(t)\|_{L^{\frac{3(m+3)}{m+6}}}$ and,
 since $1<3\,(m+3)/(m+6)<m+3$,
  we obtain
\begin{equation*}
\|E_{2,2}(t)\|_{L^{m+3}}\leq  
C\,(\|\rho_{2,2}(t)\|_{L^{m+3}}+\|\rho_{2,2}(t)\|_{L^1}).
\end{equation*}
On the other hand, by \eqref{ineq:inverse-M} and \eqref{dvp5}  (Subsection \ref{sub:appendix-flow} in the 
Appendix), we have
$M^{-1}=I_3+sP_5$ with  $\max_{i}\|\pa_{v_i} P_5\|_{L^{\infty}([0,t]\times\R^3\times\R^3)}\leq C$ and with $\|M\|_{L^{\infty}([0,t]\times\R^3\times\R^3)}\leq 2$ by \eqref{nnnun-M}, so that
\begin{equation}\label{nnnquatre} 
\begin{split}
 \| \dev_v (^t M)(s)\|_{L^\infty(\R^3\times\R^3)}\leq C s,\quad \forall s\in[0,t].
\end{split}
\end{equation}
Similarly, by \eqref{ineq:inverse-N} and \eqref{dxp6} we have $N=sP_6$ with  $\max_i \|\pa_{x_i} P_6\|_{L^{\infty}([0,t]\times\R^3\times\R^3)}\leq C$, hence
\begin{equation}\label{nnnquatre-bis} 
\begin{split}
 \|\dev_x (^tN)(s)\|_{L^{\infty}(\R^3\times\R^3)}\leq Cs,\quad \forall s\in[0,t].
\end{split}
\end{equation}
Therefore 
\begin{equation}\label{eq:rho22-1}
\|\rho_{2,2}(t)\|_{L^{m+3}}\leq 
C\left\|\int_0^t s \int_{\R^3} (|E_{\text{int}} + F_{\text{int}}|\,f)(t-s, X(s,x,v), V(s,x,v)) \, dv\, ds\right\|_{L^{m+3}}.\end{equation}
Next, we have by changing variable,
\begin{equation}\label{eq:rho22-2}
\begin{split}
\|\rho_{2,2}(t)\|_{L^1}&\leq 
C\int_0^t s \iint_{\R^3\times \R^3} 
 [(|E_{\text{int}}| + |F_{\text{int}}|)f](t-s, x, v) \,dx\, dv\, ds.
\end{split}
\end{equation}
By an argument similar to that of \eqref{ineq:pao},
 we have for $s\in[0,t]$
\begin{equation*}
\begin{split}
\iint_{\R^3\times \R^3} (|E_{\text{int}}|\,f)(s, x, v) \,dx\, dv\leq \| E_{\text{int}}(s)\|_{L^{5/2}}\,
\|\rho(s)\|_{L^{5/3}} \le C.\end{split}\end{equation*}
On the other hand, Proposition \ref{prop:virial} entails that
\begin{equation*}
\begin{split}
\int_0^t s \iint_{\R^3\times \R^3} (|F_{\text{int}}|\,f)(s, x, v) \,dx\, dv\,ds\leq C\,t(1+t).\end{split}\end{equation*}
We deduce from \eqref{eq:rho22-1}-\eqref{eq:rho22-2} and the two inequalities above the estimate for $E_{2,2}$:
\begin{equation}\label{eq:rho22-3}
\|E_{2,2}(t)\|_{L^{m+3}}\leq C\,t(1+t) +
C\left\|\int_0^t s \int_{\R^3} (|E_{\text{int}} + F_{\text{int}}|\,f)(t-s, X(s,x,v), V(s,x,v)) \, dv\, ds\right\|_{L^{m+3}}.\end{equation}
Gathering \eqref{eq:rho12} and \eqref{eq:rho22-3}, we can conclude.

\end{proof}

Next sections are devoted to the control of the
right-hand side of \eqref{eqE}.

%%%%%%%%%%%%%%%%%%%%%%%%%%%%%%%%%%%%%%%%%%%%%%%%%%%%%%%%%%%%%%%%%%%%%%%%%%%%%%%%%

\subsection{Intermediate small time estimates for the moments}
\label{subsection:small-time}

\medskip

The purpose of this paragraph is to obtain estimates for the moments on small but uniform intervals of time in $[0,T]$.

\begin{proposition}
\label{prop:duhamel} 
Let $(f,\xi)$ be a classical solution of \eqref{syst:VP}-\eqref{syst:VP-3} on $[0,T]$.
Let $t\in [0,\min(1,T)]$ and let $m\geq 3$.

For any $\gamma \in ]0,1[$, we have
\begin{equation*}\begin{split}
\left\|\int_0 ^t\ s  \int_{\R^3}(|E_{\text{int}}+F_{\text{int}}|f)(t-s,X(s),V(s)) \,dv \, ds\right\|_{L^{m+3}}
\leq C(\gamma,m)\, t^{\delta} \, H_{k}(t)^{\frac{1}{m+3}},
\end{split}\end{equation*}
where $k$ is defined by $k+3=(m+3)(1+\gamma)$ (note that $k>m$),
 and where $0<\delta <1$  is defined by $\delta=\frac{\gamma}{1+(\gamma+1)(m+3)}$.
\end{proposition}

\begin{remark}
Taking $\gamma$ small in Proposition \ref{prop:duhamel},  we realize that $k>m$ may be chosen as close as we want to $m$, therefore the estimate $\|E(t)\|_{L^{m+3}}\leq C H_m(t)^{1/(m+3)}$, which in view of Lemma \ref{lemma:est-moments} would be enough to obtain an estimate on $H_m(t)$, is close to be achieved.
\end{remark}

\begin{proof} 

\noindent \textbf{Step 1: estimate for $E_{\text{int}}$.}

By Proposition \ref{prop:norm-star}, the field $E_{\text{int}}$ belongs to $L^\infty([0,T],L^{r_1}(\R^3))$ for
all $3/2<r_1\leq 15/4$. We can proceed as in the proof of the estimates (31)-(32) and (28')-(40) of \cite{LP} to get for all $3/2<r_1\leq 15/4$
\begin{equation}\label{ineq:adapt}
\begin{split}
\left\|\int_0 ^t s\ \int_{\R^3}(|E_{\text{int}}|f)(t-s,X(s),V(s))
\,dv\, ds \right\|_{L^{m+3}}&\leq
C(r_1, m) \,t^{2-\frac{3}{r_1}} \, M_{k_1}(t)^{\frac{1}{m+3}},
\end{split}
\end{equation}where $k_1>m$ is defined by $k_1+3=(m+3)(3-3/r_1)$,
which by \eqref{ineq:energy-velo} yields
\begin{equation}\label{estimate:step1}
\begin{split}
\left\|\int_0 ^t s\ \int_{\R^3}(|E_{\text{int}}|f)(t-s,X(s),V(s))
\,dv\, ds \right\|_{L^{m+3}}&\leq
C(r_1, m) \,t^{2-\frac{3}{r_1}} \, H_{k_1}(t)^{\frac{1}{m+3}}.
\end{split}
\end{equation}

Since (31)-(32) and (28')-(40) of \cite{LP} are not exactly stated as above, 
 we provide the full proof of estimate \eqref{ineq:adapt} here for the sake of completeness.
Let us consider the quantity
\begin{equation}\label{eq:E-int-loc}
\left\|\int_0^t  s\,\int_{\R^3} (|E_{\rm int}| f)(t-s, X(s), V(s))\,dv\,ds \right\|_{L^{m+3}}\;.
\end{equation}
By adapting the argument in \cite{LP}, we bound \eqref{eq:E-int-loc} applying H\"older's inequality with exponents $r_1$ and $r_1'$ such that $\frac{1}{r_1'}=1-\frac{1}{r_1}$ (and such that $3/2<r_1\leq 15/4$): 
\begin{equation*}
 \left\|\int_0^t  s\left(\int_{\R^3} |E_{\rm int}|^{r_1}(t-s,y)|{\rm det}(D_v X)|^{-1}dy\right)^{\frac{1}{r_1}}\left(\int_{\R^3} f(t-s, X(s), V(s))\,dv\right)^{\frac{1}{r_1'}} ds \right\|_{L^{m+3}}\|f\|_{L^{\infty}}^{\frac{r_1'-1}{r_1'}}\;,
\end{equation*} 
where we perform the change of variable $X(s,x,v)\mapsto y$, with Jacobian $|\text{det}(D_v X)|^{-1}$. Thanks to estimate \eqref{nnquatre-det} in the Appendix, we can bound the above quantity by:
\begin{equation*}
\begin{split}
 &\left\|\int_0 ^t s\,\sup_{\tau\in[0,t]} \left(\int_{\R^3} |E_{\rm int}|^{r_1}(\tau,y)\,\frac{8}{s^3}\,dy\right)^{\frac{1}{r_1}}\,\left(\int_{\R^3} f(t-s, X(s), V(s))\,dv\right)^{\frac{1}{r_1'}} \:ds\right\|_{L^{m+3}}\|f\|_{L^{\infty}}^{\frac{r_1'-1}{r_1'}}\\
& \le  C\,\sup_{\tau\in [0,t]}\|E_{\rm int}(\tau)\|_{L^{r_1}} \left\|\int_0 ^t s^{1 - 3/r_1}
  \,\left(\int_{\R^3} f(t-s, X(s), V(s))\,dv\right)^{\frac{1}{r_1'}} \:ds\right\|_{L^{m+3}}\\
& \le  C\,\sup_{\tau\in [0,t]}\|E_{\rm int}(\tau)\|_{L^{r_1}} \int_0 ^t s^{1 - 3/r_1}  \left\|
  \,\left(\int_{\R^3} f(t-s, X(s), V(s))\,dv\right)^{\frac{1}{r_1'}} \right\|_{L^{m+3}} \:ds\\
&\leq C\,t^{2-\frac{3}{r_1}}\,\sup_{\tau\in [0,t]}\|E_{\rm int}(\tau)\|_{L^{r_1}}\,
\sup_{(\tau,s)\in[0,t]^2}
\left\|\int f(\tau, X(s), V(s))\,dv  \right\|_{L^{\frac{m+3}{r_1'}}}^{\frac{1}{r_1'}}\;.
\end{split}
\end{equation*}
 Applying Proposition \ref{prop:norm-star}, we obtain (since $3/2<r_1\leq 15/4$)
$$\sup_{\tau\in [0,t]}\|E_{\rm int}(\tau)\|_{L^{r_1}}\leq C(r_1).$$
 We are left with the estimate of the second factor in the above equation. We apply the formula \eqref{int2-modified}
 in Proposition \ref{interp-modified}, observing that $t\leq 1$, where $b$ is chosen in  such
a way that $b+3=\frac{3(m+3)}{r_1'}$. Note that $b\geq 0$ because $m\geq r'_1$ since $m\geq 3$ and $r'_1\leq 3$. So we get
\begin{equation*}
\sup_{(\tau,s)\in[0,t]^2}
\left\|\int_{\R^3} f(\tau, X(s), V(s))\,dv\right\|_{L^{\frac{m+3}{r_1'}}}^{\frac{1}{r_1'}}\leq C \|f_0\|_{L^{\infty}}^{\frac{b}{r'_1\,(b+3)}}\,(1+M_{k_1}(t))^{\frac{1}{m+3}}\;,
\end{equation*}
where $k_1=b$. By using eq. \eqref{ineq:energy-velo}, we obtain the bound \eqref{estimate:step1}. We observe that (since $k_1+3=\frac{3(m+3)}{r_1'}$ and $\frac{1}{r_1'}=1-\frac{1}{r_1}$),
we have the identity $k_1=(m+3)(3-3/r_1)-3$.

\medskip

\medskip

\noindent \textbf{Step 2: estimate for $F_{\text{int}}$.}

Let
\begin{equation*}
I(x)=\int_0^t s\, \int_{\R^3}(|F_{\text{int}}| f)(t-s,X(s),V(s))\,dv\, ds.
\end{equation*}
\medskip

In the following we will write $\xi$ instead of $\xi(t-s)$ when not
misleading.

\medskip

\noindent \textbf{\textbullet\: Local estimate for $I$.}

We recall that  by Proposition \ref{prop:dist}, there exists $R_b\geq 4$ 
such that $\sup_{t\in [0, \min(1,T)]}|\xi(t)|\leq R_b.$ We set $B=B(0,3R_0)$, where $R_0 := \sup(R_b, 2R(1))$ (where
$R(T)$ is defined in the Appendix, formula \eqref{cond:R}).

\medskip

Let
$0<\e<2/(m+3)$ be a small parameter and let us pick $3/(2+\e)<r_2<3/2$. By H\"older's inequality,
 we get for all $x\in B$
\begin{equation*}
\begin{split}
&I(x)\\
&= \int_0^t s^{1+\e}\,  \ints \frac{|F_{\text{int}}|(t-s,X(s))}{ |X(s)-x|^{\e}}
\,\left(\frac{|X(s)-x|}{s}\right)^{\e} f(t-s,X(s),V(s))\,dv\, ds\\
&\leq \|f_0\|_{L^\infty}^{1-\frac{1}{r'_2}}\,\int_0^t s^{1+\e}\,\left(\ints\frac{|F_{\text{int}}|^{r_2}(t-s,X(s))}{ |X(s)-x|^{\e r_2}}\,dv  \right)^{\frac{1}{r_2}}
\left(\ints  \left(\frac{|X(s)-x|}{s}\right)^{\e r_2'}
f(t-s,X(s),V(s))\,dv\right)^{\frac{1}{r_2'}} ds\;.
\end{split}
\end{equation*}
For all fixed $x\in B$ and $0\leq s\leq \min(1,T)$ we perform the change of variable $ v \mapsto y$
with 
 $$y=x-X(s,x,v).$$ 
Then, by formula \eqref{nnndeux} in the Appendix,
$|y|\leq s\, |v| + C\,\frac{s^2}{R(1)^2} \le  s\,(|v|+C)$.
Moreover we have
$$\ints \frac{|F_{\text{int}}|^{r_2}(t-s,X(s))}{ |X(s)-x|^{\e r_2}}\,dv
=\ints \frac{|F_{\text{int}}|^{r_2}(t-s,x-y)}{ |y|^{\e r_2}}
|\text{det}(D_v X(s))|^{-1}\,dy$$
and, according to estimate (\ref{nnquatre-det}) in the Appendix,
 we have $|\text{det}(D_v X(s))|^{-1}\leq 8/s^3$.

So we obtain
\begin{equation*}
\begin{split}
&\|I\|_{L^{m+3}(B)}\\
&\leq C\left\|\int_0 ^t \ s^{1-\frac{3}{r_2}+\e}\
\left(\ints \frac{|F_{\text{int}}|^{r_2}(t-s,x-y)}{ |y|^{\e r_2}}\,dy\right)^{\frac{1}{r_2}}\left(\ints
(1+ |v|^{\e r_2'})
f(t-s,X(s),V(s))\,dv\right)^{\frac{1}{r_2'}} \, ds\right\|_{L^{m+3}(B)}\\
&\leq C\int_0^t \ s^{1-\frac{3}{r_2}+\e}J(s)^{\frac{1}{m+3}}\,ds,
\end{split}\end{equation*}
where
\begin{equation*}\begin{split}
&J(s)\\&=\int_{|x|\leq 3R_0}
\left(\int_{|x-\xi-y|\leq 2R(1)}\frac{dy}{|y|^{\e
r_2}|x-\xi-y|^{2r_2}}\right)^{\frac{m+3}{r_2}}
\left(\ints (1+ |v|^{\e r_2'})
f(t-s,X(s),V(s))\,dv\right)^{\frac{(m+3)}{r_2'}} \, dx \\
&\leq C(r_2, \e) \,\int_{|x|\leq 3R_0} \
|x-\xi|^{(\frac{3}{r_2}-2-\e)(m+3)}\left(\ints
(1+|v|^{\e r_2'})
f(t-s,X(s),V(s))\, dv\right)^{\frac{(m+3)}{r_2'}}\, dx.
\end{split}
\end{equation*}
Note that in the estimate above, we used the fact that
$$ \int_{\R^3} \frac{dy}{|y|^{\e r_2}\,|x-\xi-y|^{2r_2}} = C(r_2, \e) \, 
|x-\xi|^{3 - \e r_2 - 2r_2}. $$

We now set $$r_2=\frac{3}{2+\e/2}$$ 
and we define $$p = \frac{2}{\e\,(m+3)}.$$ 
Note that $p>1$ and 
 \begin{equation}\label{cond:r2}
 \frac{3}{2+\frac{1}{m+3}}<r_2 <\frac{3}{2}.
 \end{equation}
  Moreover
\begin{equation*}%\label{cond:p}
-(\frac{3}{r_2}-2-\e)(m+3)p =\frac{\e}{2}(m+3)p = 1 <3.
\end{equation*}

Applying H\"older's inequality, we obtain
\begin{equation*}
\begin{split}
&J(s)\\&\leq C\,
\left(\int_{|x|\leq 3R_0}
|x-\xi|^{-\frac{\e}{2}(m+3)p}\,dx\right)^{1/p}
\left(\int_{|x|\leq 3R_0}  \left(\ints  (1+|v|^{\e r_2'})
f(t-s,X(s),V(s))\,dv \right)^{\frac{(m+3)}{r_2'}p'}\,dx\right)
^{1/p'} ,\end{split}\end{equation*}
therefore
\begin{equation*}
\begin{split}
\|I\|_{L^{m+3}(B)}&\leq C \, R_0^{\var}\, t^{\frac{\e}{2}}
\sup_{\tau,s\in [0,t]}\left\{\ints  \left(\ints(1+  |v|^{\e
r_2'})f(\tau,X(s),V(s))\,dv\right)^{\frac{(m+3)}{r_2'}p'}
\,dx\right\}^{\frac{1}{(m+3)p'}}.
\end{split}
\end{equation*}

We now focus on the right-hand side
\begin{equation*}
\begin{split}
\sup_{\tau,s\in[0,t]}\left\{\ints  \left(\ints (1+|v|^{\e
r_2'})f(\tau,X(s),V(s))\,dv\right)^{\frac{(m+3)}{r_2'}p'}\,dx
\right\}^{\frac{1}{(m+3)p'}}.
\end{split}
\end{equation*}

Let us introduce $k_2$ such that
\begin{equation*}
\left(\frac{3+\e r_2'}{3+k_2}\right)\left(\frac{m+3}{r_2'}p'\right)=1
\end{equation*}
and apply \eqref{estimate-moment-modified} and \eqref{int2-modified} with this choice and $a=\e r_2'$, 
$b= k_2$. Note that $b>a$ since 
$$ \frac{(m+3)}{r_2'}p' > 1 \iff (m+3) (\frac{1}{3} - \frac{\e}{6}) > 1 - \e\,\frac{m+3}{2}
\iff \frac{1}{3} > \frac{1}{m+3} - \frac{\e}{3}, $$ which holds as soon as $m>0$ (remember that 
$m\geq 3$ in this Proposition).

 We obtain
\begin{equation*}
\begin{split}
\sup_{\tau,s\in [0,t]}&
\left\{\ints \ \left(\ints \ |v|^{\e
r_2'}f(\tau,X(s),V(s))\,dv\right)^{\frac{(m+3)}{r_2'}p'}
\,dx\right\}^{\frac{1}{(m+3)p'}}\\&\leq C\sup_{\tau,s\in [0,t]}\left\{\intd \ |v|^{k_2}
f(\tau,X(s),V(s))\,dx\,dv\right\}^{\frac{1}{(m+3)p'}}.
\end{split}
\end{equation*}

Similarly, we introduce $k'_2$ such that
\begin{equation*}
\left(\frac{3}{3+k'_2}\right)\left(\frac{m+3}{r_2'}p'\right)=1,
\end{equation*}
and apply \eqref{estimate-moment-modified} and \eqref{int2-modified} with this choice and $a=0$, 
$b= k'_2>a$.
Since $k'_2<k_2$, we obtain
\begin{equation*}
\begin{split}
\sup_{\tau,s\in [0,t]}&
\left\{\ints \ \left(\ints f(\tau,X(s),V(s))\,dv\right)^{\frac{(m+3)}{r_2'}p'}
\,dx\right\}^{\frac{1}{(m+3)p'}}\\
&\leq C \sup_{\tau,s\in [0,t]}\left\{\intd \ |v|^{k'_2}
f(\tau,X(s),V(s))\,dx\,dv\right\}^{\frac{1}{(m+3)p'}}\\
&\leq C\sup_{\tau,s\in [0,t]}\left\{\intd (1+\ |v|^{k_2})
f(\tau,X(s),V(s))\,dx\,dv\right\}^{\frac{1}{(m+3)p'}}\\
&\leq C\sup_{\tau,s\in [0,t]}\left\{\intd (1+\ |V(s)|^{k_2})
f(\tau,X(s),V(s))\,dx\,dv\right\}^{\frac{1}{(m+3)p'}}\\
&\leq C+C \sup_{\tau,s\in [0,t]}\left\{\intd \ |V(s)|^{k_2}
f(\tau,X(s),V(s))\,dx\,dv\right\}^{\frac{1}{(m+3)p'}}.
\end{split}
\end{equation*}
We have used \eqref{nnntrois} together with the condition \eqref{cond:R} in the last inequalities.

Finally, using that $(X(s),V(s))$ preserves the Lebesgue's measure we obtain
\begin{equation*}
\|I\|_{L^{m+3}(B)}\leq C\, R_0^{\var}\,t^{\frac{\e}{2}}
\,H_{k_2}(t)^{\frac{1}{(m+3)p'}}.
\end{equation*}
Since $p'\geq 1$, making explicit the dependence of the constants, we get for any
$r_2$ satisfying the condition \eqref{cond:r2},
\begin{equation}
\label{estimate:I-local}
\|I\|_{L^{m+3}(B)}\leq C(r_2,m)\, R_0^{\frac{3}{2r_2} - 1}\, t^{\frac{3}{r_2}-2}
\,H_{k_2}(t)^{\frac{1}{m+3}},
\end{equation}
where  $k_2>m$ satisfies
\begin{equation}\label{def:k2}
\begin{split}
3+k_2=(3+m)
\frac{\frac{3}{r_2}-1}{1-(m+3)\left(\frac{3}{r_2}-2\right)} 
= (3+m)
\frac{1 + \e/2}{1 - (m+3)\,\e/2} . 
\end{split}
\end{equation}

\medskip

\noindent \textbf{\textbullet \: Estimate for $I$ at infinity.}

In this step we estimate the norm of $I$ on the exterior of $B=B(0,3R_0)$. Observe that when
$|x|\geq 3R_0$ and when
$|x-\xi(t-s)-y|\leq 2R(1)$, we have
$$|y|\geq |x|-|\xi(t-s)|-2R(1)\geq 2R_0 -2R(1) \geq 2R(1)\geq 1.$$  We use again the parameters $0<\e<2/(m+3)$ and
$r_2=3/(2+\e/2)$. By similar computations we find
\begin{equation*}
\begin{split}
&\|I\|_{L^{m+3}(B^c)}\\
&\leq C\left\|\int_0 ^t  s^{1-\frac{3}{r_2}+\e}\,
\left(\int_{|x-\xi-y|\leq 2R(1)}
\frac{dy}{|y|^{\e r_2}|x-\xi-y|^{2r_2}}\right)^{\frac{1}{r_2}}
\left(\ints(1+|v|^{\e r_2'})
f(t-s,X(s),V(s))\, dv\right)^{\frac{1}{r_2'}}\right\|_{L^{m+3}(B^c)} \, ds\\
&\leq C\int_0^t s^{1-\frac{3}{r_2}+\e}
\left\{\ints 
\left(\int_{|x-\xi-y|\leq 2R(1)}
\frac{dy}{|x-\xi-y|^{2r_2}}\right)^{\frac{m+3}{r_2}}
\left(\ints 
(1+|v|^{\e r_2'})
f(t-s,X(s),V(s))\,dv\right)^{\frac{(m+3)}{r_2'}}\, dx \right\}^{\frac{1}{m+3}}  
\, ds \\
&\leq
 C\, R(1)^{3/r_2 - 2}\int_0^t \ s^{1-\frac{3}{r_2}+\e}
\left\{\ints  \left(\ints (1+\ |v|^{\e
r_2'}) f(t-s,X(s),V(s))\, dv \right)^{\frac{(m+3)}{r_2'}}\, dx \right\}^{\frac{1}{m+3}}\, ds.
\end{split}
\end{equation*}
We now introduce $k_3$ such that
\begin{equation*}
\left(\frac{3+\e r_2'}{3+k_3}\right)\left(\frac{m+3}{r_2'}\right)=1.
\end{equation*}
Note that $k_3 >  \e\,r'_2$ when $m+3> r'_2$, which is 
always true if (like in our case) $m \ge 3$ (remember that $\var < \frac2{m+3}$ so that
$r' \le \frac{3}{1 - \frac1{m+3}}$).
It follows from \eqref{estimate-moment-modified} with $a= \e\,r'_2$ and 
$b = k_3$ that
\begin{equation*}
\begin{split}
\|I\|_{L^{m+3}(B^c)}\leq C\,  R(1)^{\frac{3}{r_2} - 2}\,t^{\frac{3}{r_2}-2}
\,H_{k_3}(t)^{\frac{1}{m+3}}.
\end{split}
\end{equation*}
Since $k_3<k_2$, we have $H_{k_3}\leq CH_{{k}_2}$.
Finally, making explicit the dependence of the constant:
\begin{equation}\label{estimate:I-infinity}
\begin{split}
\|I\|_{L^{m+3}(B^c)}\leq C(r_2,m)\,  R(1)^{\frac{3}{r_2} - 2}\,  t^{\frac{3}{r_2}-2}
\,H_{{k}_2}(t)^{\frac{1}{m+3}},
\end{split}
\end{equation}
with $k_2 = \frac{(3+m)\,(3/r_2 -1)}{1 - (m+3)(3/r_2 - 2)} - 3$, for all $r_2$ such that
$\frac{3}{2+ 1/(m+3)}<r_2< 3/2$.

\medskip

\noindent \textbf{Step 3: end of the proof of Proposition \ref{prop:duhamel}.}

Gathering the estimates \eqref{estimate:I-local} and \eqref{estimate:I-infinity}, we find
\begin{equation}\label{estimate:Fint}
\begin{split}
\left\| \int_0^t s \int_{\R^3} (|F_{\text{int}}|\,f) (t-s,X(s),V(s))\,dv \, ds \right\|_{L^{m+3}}
\leq  C(r_2,m) \, R_0^{\frac{3}{2r_2} - 1} \, t^{\frac{3}{r_2}-2}\, H_{k_2}(t)^{\frac{1}{m+3}},
\end{split}\end{equation}
hence (dropping the dependence w.r.t. $R_0$ and $R(1)$)
\begin{equation}\label{ineq:duhamel} \begin{split}
\left\|\int_0 ^t\ s  \int_{\R^3} (|E_{\text{int}}+F_{\text{int}}\,|f)(t-s,X(s),V(s)) \,dv\, ds \right\|_{L^{m+3}}
&\leq C(r_1,m)\, t^{2-\frac{3}{r_1}} \, H_{k_1}(t)^{\frac{1}{m+3}}\\
&+C(r_2,m)\, t^{\frac{3}{r_2}-2} \, H_{k_2}(t)^{\frac{1}{m+3}}.
\end{split}
\end{equation}
We recall that  $3/2<r_1\leq 15/4$ and $\frac{3}{2+ 1/(m+3)}<r_2< 3/2$ can be taken
 as close as necessary  to $3/2$, and that
$k_1,k_2>m$ are defined by $k_1+3=(m+3)(3-3/r_1)$ and by
$k_2+3=(m+3)(\frac{3}{r_2}-1)/(1- (m+3)\,(3/r_2-2) )$ (see \eqref{def:k2}).

We next choose $r_1$ and $r_2$ so that $k_1=k_2$ in the following way. We consider a 
small parameter $0<\gamma<1$. We define $r_1$ so that
$$2-\frac{3}{r_1}=\gamma.$$ Note that $3/2<r_1<3\leq 15/4$ by choice of $\gamma$.

We next define $r_2$ so that
$$
\frac{3}{r_2}-2=\frac{\gamma}{1+(m+3)(\gamma+1)},$$
which implies that $k_2=k_1$. Then the condition \eqref{cond:r2} on $r_2$ is satisfied.
Noticing that $k+3=(m+3)(1+\gamma)$, and using that $t\leq 1$, \eqref{ineq:duhamel} rewrites
\begin{equation*}\begin{split}
\left\|\int_0 ^t\ s\,  \int_{\R^3}(|E_{\text{int}}+F_{\text{int}}|\, f) (t-s,X(s),V(s)) \,dv\, ds \right\|_{L^{m+3}}
&\leq C(\gamma, m)\,\Big(t^\gamma+ t^{\frac{\gamma}{1+(m+3)(\gamma+1)}}\Big)\,
H_{k}(t)^{\frac{1}{m+3}}\\
&\leq C(\gamma,m)\, t^{\frac{\gamma}{1+(m+3)(\gamma+1)}} \,
H_{k}(t)^{\frac{1}{m+3}}\\
&\leq C(\gamma,m)\, t^{\delta} \,
H_{k}(t)^{\frac{1}{m+3}}.
\end{split}\end{equation*}
The conclusion follows.

\end{proof}

\begin{proposition}[Intermediate small time estimates]\label{prop:small-time}
Let $(f,\xi)$ be a classical solution of \eqref{syst:VP}-\eqref{syst:VP-3} on $[0,T]$.
For $t\leq \inf(1,T)$ and $3<m<m_0$, the following estimate holds:
\begin{equation}\label{ineq:norm-Equatre}\begin{split}
\|E(t)\|_{m+3}&
\leq C+C \left\|\int_0 ^t\ s  \int_{\R^3} (|E_{\mathrm{int}}+F_{\mathrm{int}}|\,f)(t-s,X(s),V(s))\,dv\,ds \right\|_{L^{m+3}}\\&\leq C+Ct^{\delta}
+Ct^{1+ \gamma+\delta}
H_m(t)^{\frac{3(k+3)}{(m+3)^2}}.\end{split}
\end{equation}

Here $\gamma$ is any number in $]0,1[$ such that $\gamma<\frac{m_0-m}{m+3}$ if $m\ge 6$,
 and any number in $]0,1[$ such that
\begin{equation}
\label{cond:gamma}
 \gamma\leq \frac{m-3}{6-m},\quad \gamma<\frac{m_0-m}{m+3}\end{equation} if $m<6$. 
The parameter  $k>m$ is defined by $k+3=(m+3)(1+\gamma)$, and $0<\delta<1$ by
 $$\delta=\frac{\gamma}{1+(m+3)(\gamma+1)}$$

\end{proposition}

\begin{remark}
We stress that the constants depend on $k$, or equivalently, on $\gamma$
 (in fact some of them blow up
when $k\to m$).
\end{remark}
\begin{proof}
Thanks to Propositions \ref{prop:duhamel-0} and \ref{prop:duhamel}, we obtain (remember that $t<1$)
\begin{equation}\label{ineq:o1}\begin{split}
\|E(t)\|_{L^{m+3}}&\leq C + 
C \left\|\int_0 ^t\ s\,
  \int_{\R^3} (|E_{\text{int}}+F_{\text{int}}|\,f)(t-s,X(s),V(s))\,dv \, ds \right\|_{L^{m+3}}\\
&\leq C+ C t^{\delta}  \,  H_{k}(t)^{\frac{1}{m+3}},\end{split}
\end{equation}
[with $k>m$ such that $k+3=(m+3)(1+\gamma)$ and 
$\delta=\frac{\gamma}{1+(m+3)(\gamma+1)}$, and for all $0<\gamma<1$].

On the other hand, we infer from Lemma \ref{lemma:est-moments} and from Proposition \ref{prop:virial} that
\begin{equation*}
H_k(t)^{\frac{1}{m+3}}\leq C\left(H_k(0)^{\frac{1}{m+3}}+
t^{\frac{k+3}{m+3}}\sup_{s\in
[0,t]}\|E(s)\|_{L^{k+3}}^{\frac{k+3}{m+3}}+1\right).
\end{equation*}
Therefore, since $k<m_0$ because of the second condition on $\gamma$, we get
\begin{equation}\label{ineq:o2}
H_k(t)^{\frac{1}{m+3}}\leq C\left(1+ t^{\frac{k+3}{m+3}}\sup_{s\in
[0,t]}\|E(s)\|_{L^{k+3}}^{\frac{k+3}{m+3}}\right).
\end{equation}

Next,  combining Propositions \ref{interp} and \ref{sobo} we have
$\|E(s)\|_{L^{3(m+3)/(6-m)}}\leq CH_m(s)^{\frac{3}{m+3}}$ if $m<6$, and 
$\|E(s)\|_{L^{\infty}} \leq CH_m(s)^{\frac{3}{m+3}}$ if $m>6$.
Observing that $k+3\leq 3(m+3)/(6-m)$ by the first condition of \eqref{cond:gamma},
this yields
\begin{equation}\label{ineq:interpol}
\|E(s)\|_{L^{k+3}}\leq C H_m(s)^{\frac{3}{3+m}}.
\end{equation}

Therefore, we infer from \eqref{ineq:o1}, \eqref{ineq:o2} and
\eqref{ineq:interpol}, that
\begin{equation}\label{ineq:norm-Eseconde}\begin{split}
\|E(t)\|_{m+3}&\leq C+ C\left\|\int_0 ^t\ s  \int_{\R^3}(|E_{\text{int}}+F_{\text{int}}|\,f)(t-s,X(s),V(s))\,dv ds \right\|_{L^{m+3}} \\
&\leq C+Ct^{\delta}
+Ct^{1+ \gamma+\delta}
H_m(t)^{\frac{3(k+3)}{(m+3)^2}}.\end{split}
\end{equation}

This completes the
Proof of Proposition \ref{prop:small-time}.

\end{proof}

%%%%%%%%%%%%%%%%%%%%%%%%%%%%%%%%%%%%%%%%%%%%%%%%%%%%%%%%%%%%%%%%%%%%%%%%%%%%%%%%%%%%%%

\subsection{Bound on the moments}
\label{subsection:large-time}

This paragraph is devoted to the proof of the propagation of the moments, formulated in the following 
\begin{proposition}
\label{prop:global-estimate}Let $(f,\xi)$ be a classical solution of \eqref{syst:VP}-\eqref{syst:VP-3} on $[0,T]$.
Then we have for all $0\leq m<\min(m_0,7)$
\begin{equation} \label{estipol}
H_m(T)\leq C(1+T)^{c},
\end{equation}
where $C$ only depends on the quantities $\mathcal{H}(0)$, $\mathcal{M}(0)$, 
$\|f_0\|_{\infty}$, $\xi_0$, $m_0$, $H_m(0)$ for $m<m_0$; and $c$ only depends on $m$.
\end{proposition}

\begin{remark}\label{rk:explicit-const}
The constant $c$ can be estimated in terms of $m$. Indeed, it is possible to take in (\ref{estipol})
any real number stricly bigger than
$$ c_0(m) := \frac{31}{10}\, \min_{\gamma \in ]0,1[} \max \bigg\{ \frac{m+3}{\gamma} , \frac52\,
\left[ \frac1{m-2} - \frac{2+4\,\gamma}{(m+3)\,(1 + \gamma + \gamma/(1 + (\gamma + 1)\,(m+3) ))} 
\right]^{-1}  \bigg\}. $$
One can check that $c_0(6)$ is of order $500$, while $c_0(m) \to \infty$ when $m \to 7$.

\end{remark}

\begin{proof}

\medskip

{We begin with the case $\frac{16}{3}< m <\min(m_0,7)$}.

Let $t\in[0,T]$. In view of Lemma \ref{lemma:est-moments} and Proposition \ref{prop:duhamel-0} it is enough to control the quantity 
\begin{equation*}
\left\|\int_0 ^t\ s
\int_{\R^3} (|E_{\text{int}}+F_{\text{int}}|\,f) (t-s,X(s),V(s))\,dv \, ds\right\|_{L^{m+3}}
\end{equation*}
in terms of $H_m^{1/(m+3)}$.
Unfortunately, the bound obtained in Proposition \ref{prop:small-time} does not allow to conclude, since it provides an exponent $3(k+3)/(m+3)^2$, which is much too large. In order to bypass this difficulty, we shall use, as in \cite{LP}, two kinds of estimates: for small times we will use the estimate of Proposition \ref{prop:small-time}; note indeed that the right-hand side is small when $t$ is small.  On the other hand for large times we will perform other estimates.

More precisely, let $0<t_0<\inf(1,T)$  be sufficiently small, to be
determined later on. Let $\alpha \in ]0,1/4[$.

\medskip

\textbf{First case: $t\in [t_0,T].$}

We have thanks to \eqref{nnquatre-det},
\begin{equation*}\begin{split}
& \left\|
 \int_{t_0}^t s \ints (|F_{\text{int}}|\,f)(t-s,X(s),V(s))\, dv \,ds
\right\|_{L^{m+3}} \\
& \le \left\| \int_{t_0}^t s \bigg\{ \ints |F_{\text{int}}|^{3/2 - \alpha}(t-s,X(s)) \,
dv \bigg\}^{\frac2{3-2\alpha}}
 \bigg\{ \ints f^{\frac{3 - 2\,\alpha}{1 -
2\,\alpha}}(t-s,X(s),V(s)) \, dv \bigg\}^{\frac{1 - 2\,\alpha}{3 - 2\,\alpha}} \, ds
\right\|_{L^{m+3}} \\
&\leq
C \int_{t_0}^t s \bigg\{ \ints \chi_R(y)\frac{1}{|y|^{3 -2 \alpha}} \,
\frac{dy}{s^3} \bigg\}^{\frac2{3-2\alpha}}   \,ds\,
  \sup_{\tau,\tau'\in[0,t]} \left\| \ints f^{\frac{3 - 2\,\alpha}{1 -
2\,\alpha}}(\tau,X(\tau'),V(\tau')) \, dv  \right\|_{L^{(m+3)\,\frac{1-2\,\alpha}{3 -
2\,\alpha}}}
 ^{\frac{1 - 2\,\alpha}{3 - 2\,\alpha}}  \\
&\leq
 C \, R(T)^{\frac{4\alpha}{3 - 2\alpha}} \,\|f_0\|_{L^{\infty}}^{\frac2{3-2\alpha}}\int_{t_0}^t s^{1 -
\frac{6}{3-2\,\alpha}} 
\,ds\,
\sup_{\tau,\tau'\in[0,t]} \left\| \ints f(\tau,X(\tau'),V(\tau')) \, dv
\right\|_{L^{(m+3)\,\frac{1-2\,\alpha}{3 - 2\,\alpha}}}^{\frac{1 -
2\,\alpha}{3 - 2\,\alpha}}.
\end{split}\end{equation*}

We now use Proposition \ref{interp-modified}, more precisely estimate \eqref{estimate-moment-modified},
 with $a=0$ and $b$ such that
$\frac{b+3}3=(m+3)\, \frac{1-2\,\alpha}{3 - 2\,\alpha}$. Note that
$b>0$ since $\alpha \in ]0,1/4[$ and $m>2$. We obtain
\begin{equation*}\begin{split}
 \left\|
 \int_{t_0}^t s \ints (|F_{\text{int}}|\,f)(t-s,X(s),V(s))\, dv \,ds
\right\|_{L^{m+3}} 
&\leq
 C  \, R(T)^{\frac{4\alpha}{3 - 2\alpha}} \, t_0^{-
\frac{4\alpha}{3-2\alpha}}
\, H_{\frac{m - \alpha\,(2m+4)}{1 - \frac23\,\alpha}}(t)^{\frac1{m+3}}.
\end{split}
\end{equation*}

We now apply the interpolation inequality
\begin{equation}\label{eq:energy-interpolation}
 H_{\beta} (t)\leq H_2(t)^{\frac{m-\beta}{m-2}} H_m(t)^{\frac{\beta -
2}{m-2}},\quad \beta \in[2,m],
 \end{equation}
 with the choice $\beta=\frac{m - \alpha\,(2m+4)}{1 - \frac23\,\alpha}$.
 Note that $\beta \in [2,m]$ since {$m > 16/3$} and $\alpha < 1/4$.
Since $\sup_{t \in [0,T]} H_2(t) \leq C$ thanks to the conservation
of energy, this yields
\begin{equation}
\label{ineq:close}
\begin{split}
 \left\|
 \int_{t_0}^t s \ints (|F_{\text{int}}|\,f)(t-s,X(s),V(s))\, dv \,ds
\right\|_{L^{m+3}} \leq
 C \,R(T)^{\frac{4\alpha}{3 - 2\alpha}} \, t_0^{- \frac{4\alpha}{3-2\alpha}}
\, H_{m}(t)^{\frac1{m+3}-\frac{4\alpha}{(3-2\alpha)(m-2)}}.
\end{split}\end{equation}

\medskip

%%%%%%%%%%%%%%%%%%%%%%%%%%%%%%
%%%%%%%%%%%%%%%%%%%%%%%%%%%%%%

We obtain an analogous estimate for the internal part of the electric field. 
Since $\rho$ belongs to $L^\infty([0,T],L^{5/3}(\R^3))$,
 the internal part $E_{\text{int}}$ is bounded in $L^\infty([0,T],L^{3/2-\alpha}(\R^3))$ for all $0< \alpha<1/4$ (thanks to a direct convolution inequality). So by exactly the same computations as before we find
\begin{equation}\label{ineq:close-2}
\begin{split}
&\left \|\int_{t_0}^{t} s \ints  (|E_{\text{int}}|\,f)(t-s,X(s),V(s))\,dv\,ds\right\|_{L^{m+3}}\\
&\leq\int_{t_0}^t  s^{1-\frac{6}{(3-2\alpha)}}\,ds
\sup_{\tau,\tau'\in[0,t]}
\left(\|E_{\text{int}}(\tau)\|_{3/2-\alpha}\left\|\ints  f(\tau,X(\tau'),V(\tau'))\,dv\right\|_{L^{(m+3)\left(\frac{1-2\alpha}{3-2\alpha}\right)}}^{\frac{1-2\alpha}{3-2\alpha}}\|f_0\|_{L^\infty}^{\frac{2}{3-2\alpha}}\right)\\
&\leq C \, R(T)^{\frac{4\alpha}{3 - 2\alpha}} \, t_0^{- \frac{4\alpha}{3-2\alpha}}
\, H_{m}(t)^{\frac1{m+3}-\frac{4\alpha}{(3-2\alpha)(m-2)}}.
\end{split}
\end{equation}

Combining \eqref{ineq:close} and \eqref{ineq:close-2},
 we are led to (for all 
$\alpha \in ]0, 1/4[$ and $m\in ]16/3, \min(m_0,7)[$),
\begin{equation}
\label{ineq:totale}
\begin{split}
& \left\|
 \int_{t_0}^t s \ints (|E_{\text{int}}+F_{\text{int}}|\,f )(t-s,X(s),V(s)) \, dv \,ds \right\|_{L^{m+3}}\\
&\quad\leq  C  \,  R(T)^{\frac{4\alpha}{3 - 2\alpha}}\, t_0^{- \frac{4\alpha}{3-2\alpha}}
\, H_{m}(t)^{\frac1{m+3}-\frac{4\alpha}{(3-2\alpha)(m-2)}}.
 \end{split}\end{equation}

\medskip

\textbf{Second case: $t\in [0,t_0]$.}
\medskip

By Proposition \ref{prop:small-time}, we have
 \begin{equation}\label{ihp}
 \begin{split}
& \left\|
 \int_{0}^t s \ints (|E_{\text{int}}+F_{\text{int}}|\,f)(t-s,X(s),V(s)) \, dv \,ds \right\|_{L^{m+3}}\\
&\quad \quad
\leq C\, \left(1 + t^{\delta}
+t^{1+ \gamma+\delta}
H_m(t)^{\frac{3(k+3)}{(m+3)^2}}\right),\end{split}
\end{equation}
 where $k+3=(m+3)(1+\gamma)$ and where
$\delta=\frac{\gamma}{1+(m+3)(\gamma+1)}$, with any
$\gamma \in ]0,1[$ such that \eqref{cond:gamma} is satisfied if $m<6$ and such that $\gamma < \frac{m_0-m}{m+3 }$ if $m\geq 6$. In particular,  
 \begin{equation}\label{ineq:totale-2}\begin{split}
& \left\|
 \int_{0}^{t_0} s \ints (|E_{\text{int}}+F_{\text{int}}|\,f)(t-s,X(s),V(s)) \, dv \,ds \right\|_{L^{m+3}}\\
&\quad \quad
\leq C\, \left(1 + t_0^{\delta}
+t_0^{1+ \gamma+\delta}
H_m(t_0)^{\frac{3(k+3)}{(m+3)^2}}\right)\;.\end{split}
\end{equation}

\medskip

We deduce from \eqref{ineq:totale}, \eqref{ineq:totale-2} and Proposition \ref{prop:duhamel-0}
  that if $t\in[t_0,T]$, 
\begin{equation*}
\begin{split}
 \|E(t)&\|_{L^{m+3}}\leq  
C\,(1+T)^{21/10}  + C \left\|\int_0 ^{t_0}\ s
\int_{\R^3} (|E_{\text{int}}+F_{\text{int}}|\,f) (t-s,X(s),V(s))\,dv \, ds\right\|_{L^{m+3}}\\
&+C \left\|\int_{t_0} ^t\ s
\int_{\R^3} (|E_{\text{int}}+F_{\text{int}}|\,f) (t-s,X(s),V(s))\,dv \, ds\right\|_{L^{m+3}}\\
&\leq C\,(1+T)^{21/10}+C\,t_0^\delta+C\,t_0^{1+ \gamma+\delta}
H_m(t_0)^{\frac{3(k+3)}{(m+3)^2}}
+C\,  R(T)^{\frac{4\alpha}{3 - 2\alpha}} \,t_0^{- \frac{4\alpha}{3-2\alpha}}
\, H_{m}(t)^{\frac{1}{m+3}-\frac{4\alpha}{(3-2\alpha)(m-2)}}\\
&\leq C\,(1+T)^{21/10}+C\, t_0^\delta+Ct_0^{1+ \gamma+\delta}
H_m(t)^{\frac{3(k+3)}{(m+3)^2}}+
  C   \,  R(T)^{\frac{4\alpha}{3 - 2\alpha}} \,t_0^{- \frac{4\alpha}{3-2\alpha}}
\, H_{m}(t)^{\frac{1}{m+3}-\frac{4\alpha}{(3-2\alpha)(m-2)}}.
 \end{split}\end{equation*}

On the other hand, if $t\in[0,t_0]$, \eqref{ihp} yields 
\begin{equation*}
\begin{split}
 \|E(t)\|_{L^{m+3}}&\leq  
C\,(1+T)^{21/10}  + C \left\|\int_0 ^{t}\ s
\int_{\R^3} (|E_{\text{int}}+F_{\text{int}}|\,f)(t-s,X(s),V(s))\,dv \, ds\right\|_{L^{m+3}}\\
&\leq C\,(1+T)^{21/10}+C\, t_0^\delta+Ct_0^{1+ \gamma+\delta}
H_m(t)^{\frac{3(k+3)}{(m+3)^2}}\;.
 \end{split}\end{equation*}
 
 In all cases, for $t\in[0,T]$, we have 
 \begin{equation}\label{ineq:ihp}
\begin{split}
 \|E(t)\|_{L^{m+3}}&  \leq C\,(1+T)^{21/10}+C\, t_0^\delta+Ct_0^{1+ \gamma+\delta}
H_m(t)^{\frac{3(k+3)}{(m+3)^2}}\\
&+
  C   \,  R(T)^{\frac{4\alpha}{3 - 2\alpha}} \,t_0^{- \frac{4\alpha}{3-2\alpha}}
\, H_{m}(t)^{\frac{1}{m+3}-\frac{4\alpha}{(3-2\alpha)(m-2)}}.
 \end{split}\end{equation}

We then specify our choice for $t_0$: we set (for example)
$$ t_0^{ 1 +\gamma+\delta }= H_m(t)^{-  \frac{3(k+3)-(m+3)(1-\gamma)}{(m+3)^2} }$$
hence
 \begin{equation}
 \label{def:t}
 t_0=H_m(t)^{-\frac{3(k+3)-(m+3)(1-\gamma)}{(1+\gamma+\delta)(m+3)^2}} \le 1\,. 
\end{equation}

It follows that
\begin{equation*}
t_0^{- \frac{4\alpha}{3-2\alpha}}
\, H_{m}(t)^{-\frac{4\alpha}{(3-2\alpha)(m-2)}}
=H_m(t)^{\frac{4\alpha}{3-2\alpha}\,e(m)},
\end{equation*}
where we denote by $e(m)$ the term appearing in the exponent above,
$$
e(m)=\frac{3(k+3)-(m+3)(1-\gamma)}{(1+\delta+\gamma)(m+3)^2}-\frac{1}{m-2}=
\frac{2+4\gamma}{(1+\delta+\gamma)(m+3)}-\frac{1}{m-2}.
$$
We recall that $0<\gamma<1$ is a parameter that can be chosen as small as wanted such that \eqref{cond:gamma} is satisfied if $m<6$ and such that $\gamma < \frac{m_0-m}{m+3 }$ if $m\geq 6$, 
and that $\delta=\gamma/(1+(\gamma+1)(m+3))$.

\medskip

We now use the assumption $m<7$.  Then 
$$\frac{2}{m+3}-\frac{1}{m-2}<0,$$
and  we can choose $\gamma>0$ sufficiently small, so that $e(m)< 0$ and condition \eqref{cond:gamma} is satisfied.

\medskip

Next, invoking Lemma \ref{lemma:est-moments} we have
\begin{equation*}
\label{ineq:totale-3}
\frac{d}{dt} \tilde{H}_m(t) \leq C\, \big(\|E(t)\|_{L^{m+3}}H_m^{-\frac{1}{m+3}}+|E(t,\xi(t))|H_m^{-\frac{1}{m+3}}\big)\, H_m(t)\;.
\end{equation*}
Therefore using \eqref{ineq:ihp} and our choice \eqref{def:t} of $t_0$, we obtain, since $4\alpha/(3-2\alpha)<21/10$ for $\alpha<1/4$,
\begin{equation*}
\begin{split}
\label{ineq:totale-4}
\frac{d}{dt} \tilde{H}_m(t) 
&\leq C\, \Big( (1+T)^{\frac{21}{10}}\,H_m^{-\frac{1}{m+3}} + H_m^{-\frac{\gamma}{m+3}}+ R(T)^{\frac{4\alpha}{3-2\alpha}}\,H_m^{\frac{4\alpha}{3-2\alpha}\,e(m)} + |E(t,\xi(t))|\,H_m^{-\frac{1}{m+3}}\Big)\, H_m(t)\\
&\leq C\Big( (1+T)^{\frac{21}{10}}\,H_m^{-\mbox{min}\{\frac{\gamma}{m+3}\,,\,\frac{4\alpha}{3-2\alpha}|e(m)|\}} + |E(t,\xi(t))|\,H_m^{-\frac{1}{m+3}}\Big)\,H_m(t).
\end{split}
\end{equation*}
Setting $$a=\mbox{min}\left\{\frac{\gamma}{m+3}\,,\,\frac{4\alpha}{3-2\alpha}|e(m)|\right\}\in(0,1),$$ we then have for  all $t\in[0,T]$
\begin{equation*}\begin{split}
 \frac{d}{dt} \tilde{H}_m &\leq C\Big((1+T)^{\frac{21}{10}}+|E(t,\xi(t))| \Big)\, H_m(t)^{1-a}\;.
\end{split}\end{equation*}
\medskip
Since $H_m(0)$ is bounded, by a Gronwall argument using Proposition \ref{prop:virial}, we conclude the Proof of Proposition \ref{prop:global-estimate} in the case $\frac{16}{3}<m<\min(m_0,7)$ by choosing $c=a^{-1}(\frac{21}{10}+1)$.

\medskip

Finally, by interpolation, we have for all $m\leq \frac{16}{3}$
\begin{equation*}
H_m(T)\leq C(1+H_{m'}(T))
\end{equation*}
for some $m'\in(\frac{16}{3},\min(m_0,7))$, and we conclude thanks to the previous step.

\end{proof}

\subsection{Proof of Theorem \ref{thm:main}}
\label{sec:passing-limit}
\newcommand{\eps}{\varepsilon}

We complete here the proof of Theorem \ref{thm:main}. Let $T>0$. We first consider a sequence of mollifiers $f_{0,\varepsilon}$ of $f_0$ that moreover vanish in a small $\varepsilon$-neighborhood of $\xi_0$. We then consider the unique  classical solution $(f_\varepsilon(t),\xi_\varepsilon(t))$ of \eqref{syst:VP}-\eqref{syst:VP-2} on $[0,T]$ provided by \cite{MMP}. We infer from the uniform estimates derived in Proposition \ref{prop:global-estimate} for moments of order larger than $6$ and from Proposition \ref{sobo-modified} (\eqref{ineq:bound-field-modified} with $s=0$ and $\tau=t$) that $(E_\eps)_{0<\eps<1}$ is uniformly bounded in $L^\infty([0,T]\times \R^3)$. By Ascoli's theorem there exists a subsequence such that $((\xi_{\varepsilon_n},\eta_{\varepsilon_n}))_{n\in\mathbb{N}}$ converges to some $(\xi,\eta)\in C^1([0,T])$ as $n\to \infty$, where 
$\dot{\xi}(t)=\eta(t)$.

On the other hand, by usual compactness arguments using the bounds for suitable norms of $f_{\varepsilon}$ and $\partial_t \rho_{\varepsilon}$ and on the moments, we can pass to the limit to obtain a weak solution $f$ to \eqref{syst:VP}, with $f\in C(\R_+,L^p(\R^3\times \R^3))\cap L^\infty(\R_+\times \R^3\times \R^3)$, and with the same polynomial bounds on the moments. In particular (extracting again if necessary), $(E_{\eps_n})_{n\in\mathbb{N}}$ converges to $E=\rho\ast (x \mapsto x/|x|^3)$ uniformly on $[0,T]\times \R^3$. Therefore we finally obtain $\dot{\eta}(t)=E(t,\xi(t)), t\in[0,T]$, and we conclude that $(f,\xi,\eta)$ satisfies all the conclusions of Theorem \ref{thm:main}.

%%%%%%%%%%%%%%%%%%%%%%%%%%%%%%%%%%%%%%%%%%%%%%%%%%%%%%%%%

\section{Appendix: estimates for the flow}

In the following, for a given map $F:(x,v)\in \R^3\times \R^3\mapsto F(x,v)\in \R^3$,
 we denote by $(D_x F)(x,v)$, $(D_v F)(x,v)$, the Jacobian matrices (belonging to $\mathcal{M}_3(\R)$) with respect to the variables $x$ and $v$, or, for short, $D_xF$ or $D_v F$ when not misleading, and for $F:(x,v)\in \R^3\times \R^3\mapsto F(x,v)\in \R$ we denote by $D^2_x F$ and $D^2_v F$  the Hessian matrices of $F$.

\subsection{Some estimates for the flow}\label{sub:appendix-flow}
The purpose of this paragraph is to collect a few estimates for the flow defined in 
\eqref{characteristics-2}. We consider $t \in [0,T]$.

\medskip

First we can compare $(X(s),V(s))$ with the free flow by using \eqref{estimate:Eext-0}:
\begin{equation}\label{nnndeux} 
 |  X(s,x,v) - (x-vs) |
\le \frac{(1+\|f_0\|_{L^1})}{R^2} \left(\frac{ s^2}{2}\right),\quad \forall(s,x,v)\in [0,t]\times \R^3\times \R^3,
\end{equation}
and
\begin{equation}\label{nnntrois} 
 |  V(s,x,v) - v |_{} 
\le \frac{(1+\|f_0\|_{L^1})}{R^2}\, s\quad \forall(s,x,v)\in [0,t]\times \R^3\times \R^3.
\end{equation}

From now on we choose $R:= R(T)$ large with respect to $T$ in the following way:
\begin{equation}\label{cond:R}
24\,K_0(1+\|f_0\|_{L^1})\, (1+T)^3 = R(T)^3,
\end{equation}
where $K_0 \ge 100$.
Next, differentiating \eqref{characteristics-2} with respect to $v$ we  have for $s\in[0,t]$ and $(x,v)\in \R^3\times \R^3$,
\begin{equation*}
\begin{cases}
\displaystyle
 \frac{d}{ds} (D_v X)(s,x,v) =- (D_v V)(s,x,v), \quad (D_v X)(0,x,v) = 0 \\
\displaystyle \frac{d}{ds} (D_v V)(s,x,v) = -D_x( E_{\text{ext}}
+ F_{\text{ext}})(t-s,X(s,x,v))\,D_v X(s,x,v), \quad  (D_v V)(0,x,v) = I_3.
\end{cases}
\end{equation*}
Integrating the system above and using  \eqref{estimate:Eext-1} and the definition \eqref{cond:R} of $R =R(T)$ we obtain
\begin{equation} \label{nun}
\begin{split}
\|D_v &X\|_{L^{\infty}([0,t]\times\R^3\times\R^3)}+ 
\|D_v V\|_{L^{\infty}([0,t]\times\R^3\times\R^3)}
\\&\le  T\, \exp \left( \frac{12\,(1 + \|f_0\|_{L^1})\,T}{R^3} \right) + 1 +
\frac{12\,(1 + \|f_0\|_{L^1})\,T}{R^3} \, T^2 \,\exp \left( \frac{12\,(1 + \|f_0\|_{L^1})\,T}{R^3} \right) \\&\le  4\,(1+T).
\end{split}\end{equation}
The same argument used for spatial derivatives ensures that
\begin{equation}\label{ndeux}
\|D_x X\|_{L^{\infty}([0,t]\times\R^3\times\R^3)}+ 
\|D_x V\|_{L^{\infty}([0,t]\times\R^3\times\R^3)} \le 4\,(1+T).
\end{equation}
Then, using \eqref{nun}-\eqref{ndeux} and \eqref{estimate:Eext-1}, for any $s \in [0,t]$ 
(remembering \eqref{cond:R}), we get
\begin{equation}\label{ineq:DXX}\begin{split}D_x X(s,x,v) &=I_3+\int_0^s \int_0^\sigma D_x( E_{\text{ext}}
+ F_{\text{ext}})(t-\tau,X(\tau))   \, D_x X(\tau)\,d\tau\, d\sigma\\&:=I_3+P_1(s,x,v), \end{split}
\end{equation}
where 
\begin{equation}\label{nntrois} \begin{split}
\|P_1\|_{L^\infty(\R^3\times \R^3)}
 \le 24\,(1+T)\,\frac{1+\|f_0\|_{L^1}}{R^3}\, T^2\le \frac{1}{K_0}.\end{split}
\end{equation}

Similarly,
\begin{equation}\label{ineq:DVX}\begin{split}
D_v X(s,x,v) &=s I_3+\int_0^s\int_0^\sigma D_x( E_{\text{ext}}
+ F_{\text{ext}})(t-\tau,X(\tau))   \, D_v X(\tau)\, d\tau\, d\sigma\\
&:=s\big(I_3+P_2(s,x,v)\big), 
\end{split}
\end{equation}
where
\begin{equation}\label{nnquatre} \begin{split}
\|P_2\|_{L^\infty(\R^3\times \R^3)}& \le 24\,(1+T)\,\frac{1+\|f_0\|_{L^1}}{R^3}T\leq \frac{1}{K_0}. 
\end{split}
\end{equation}
In the same way, using the definition of $R$ we obtain
\begin{equation}\label{nnun}\begin{split}
D_xV(s,x,v)&=sP_3(s,x,v),\\
 \|P_3\|_{L^\infty(\R^3\times \R^3)}&\leq
 24\,(1+T)\,\frac{1+\|f_0\|_{L^1}}{R^3}\leq \frac{1}{K_0}\min\left(1,\frac{1}{ T^2}\right)
\end{split}
\end{equation}
and
\begin{equation}\label{nndeux} \begin{split}
D_v V(s,x,v)&=I_3+sP_4(s,x,v),\quad 
 \|P_4\|_{L^\infty(\R^3\times \R^3)}\leq \frac{1}{K_0}\min\left(1,\frac{1}{ T^2}\right).
\end{split}
\end{equation}
We now fix $s\in(0,t]$ and $(x,v)\in \R^3\times \R^3$ and we omit the dependence w.r.t. $(s,x,v)$ in the matrices. 
We deduce from \eqref{ineq:DXX} and \eqref{nntrois} that (since $1/K_0<1$) $D_x X(s,x,v)$ is invertible and (since $K_0>2$),
\begin{equation}
\label{2013}
|(D_x X)^{-1}|\leq \frac{1}{1-|P_1|}\leq  2.
\end{equation}

We now recall that $K_0 \ge 100$, so that 
\begin{equation}\label{choice:K_0}|\text{det}(I_3+P)|\geq \frac{1}{2}\quad 
\text{for any }|P|\leq \frac{10}{K_0}\, \bigg( \le \frac1{10} \bigg).\end{equation}
Hence by \eqref{nnquatre}  $D_v X(s,x,v)$ is invertible  and 
\begin{equation}\label{nnquatre-det} 
 |\text{det}(D_vX)|^{-1}=s^{-3}|\text{det}(I_3+P_2)|^{-1}\leq {8}{s^{-3}}.
\end{equation}

For later purposes (see the subsection hereafter) we next show that $D_v V-(D_x V)\,(D_x X)^{-1} (D_v X)$ is invertible. Combining \eqref{ineq:DXX} -\eqref{nndeux}, we obtain
\begin{equation}\label{ineq:inverse-M}
D_v V-(D_x V)\,(D_x X)^{-1} (D_v X)=I_3+sP_5, \quad 
|P_5|_{}\leq \frac{1}{K_0}\min\left(1,\frac{1}{T^2}\right)+\frac{4T}{K_0}\min\left(1,\frac{1}{T^2}\right).
\end{equation}
Hence $T\,|P_5|<5/K_0$, so that
$(D_vV - (D_x V)\,(D_x X)^{-1} (D_v X))(s,x,v)$ is indeed invertible and 
\begin{equation}\label{nnnun-M} 
\Big\|\big[D_vV - (D_x V)\,(D_x X)^{-1} (D_v X)\big]^{-1}\Big\|_{L^\infty([0,t]\times \R^3\times \R^3)}\leq
\frac{1}{1-\|TP_5\|_{L^\infty([0,t]\times \R^3\times \R^3)}}\leq 2.
\end{equation}
On the other hand, we will also need the following estimate, which is obtained by using  \eqref{ineq:DXX} -- \eqref{ineq:DVX} and \eqref{nnnun-M}:
\begin{equation}\label{ineq:inverse-N}
(D_x X)^{-1}(D_v X) \big[D_vV - (D_x V)\,(D_x X)^{-1} (D_v X)\big]^{-1}=s P_6,\quad |P_6|\leq 8,
\end{equation}so that
\begin{equation}\label{nnnun} 
\Big\|(D_x X)^{-1}(D_v X) \big[D_vV - (D_x V)\,(D_x X)^{-1} (D_v X)\big]^{-1}\Big\|_{ L^\infty([0,s]\times \R^3\times \R^3)}\leq 8s .
\end{equation}

Finally, writing the differential system satisfied by the second-order derivatives
with respect to $x,v$ of $(X,V)$, and using \eqref{estimate:Eext-2}, we obtain 
\begin{equation}\label{ineq:derivative-2}
\max_{j,k=1,2,3}\Big(\|\pa_{(x,v)_j (x,v)_k}^2  X\|_{L^{\infty}([0,t]\times\R^3\times\R^3)}+
\|\pa_{(x,v)_j(x,v)_k}^2 V\|_{L^{\infty}([0,t]\times\R^3\times\R^3)}\Big) \le 2. 
\end{equation}
Indeed, differentiating again \eqref{characteristics-2} we have
\begin{equation*}
\begin{split}
\pa_{v_j v_k} V^{(i)}(s) &
= -\int_0^s \big[D_x^2(E_{\text{ext}}^{(i)}
+ F_{\text{ext}}^{(i)})(t-\tau,X(\tau))\pa_{v_j} X(\tau)\big]\cdot \pa_{v_k}X(\tau) \, d\tau
\\
& + \int_0^s\int_0^\tau \nabla (E_{\text{ext}}^{(i)}
+ F_{\text{ext}}^{(i)})(t-\tau,X(\tau))\cdot \partial_{v_j v_k} V(\sigma)\,d\sigma\,d\tau 
\end{split}
\end{equation*}
so that
$$ | \pa_{v_jv_k} V(s) | \le 16\, (1+T)^2 \, T\, \frac{120\,(1 + \|f_0\|_{L^1})\,T}{R^4} 
\, \exp \left(  \frac{12\,(1 + \|f_0\|_{L^1})\,T}{R^3} \right) \le 1,$$
and the other terms can be treated in the same way.

Then, we observe that
\begin{equation*}
\begin{split}
\pa_{v_i} P_4(s)=-\frac{1}{s}  \int_0^s  \pa_{v_i}\Big(D_x( E_{\text{ext}}
+ F_{\text{ext}})(t-\tau,X(\tau))   \, D_v X(\tau)\Big)\,d\tau\, 
\end{split}
\end{equation*}
so that in view of the previous estimates, we obtain
$$\max_{i=1,2,3}\|\pa_{v_i} P_4\|_{L^{\infty}([0,t]\times\R^3\times\R^3)} \le 2.$$

On the other hand, since
$$  P_2(s) = \frac1s\,\int_0^s\int_0^\sigma D_x( E_{\text{ext}}
+ F_{\text{ext}})(t-\tau,X(\tau))   \, D_v X(\tau)\, d\tau\, d\sigma, $$
by the same arguments, we get 
\begin{equation}\label{ineq:P2}\max_{i=1,2,3}\Big(\|\pa_{x_i} P_2\|_{L^{\infty}([0,t]\times\R^3\times\R^3)}+\|\pa_{v_i} P_2\|_{L^{\infty}([0,t]\times\R^3\times\R^3)}\Big) \le 2.\end{equation}

Next, recall that
$$P_5=P_4-(D_x V)(D_x X)^{-1}(I_3+P_2),$$
therefore
\begin{equation*}\begin{split}
\pa_{v_i} P_5 
&= \pa_{v_i} P_4 - \pa_{v_i}(D_{x} V) \, (D_x X)^{-1} \, (I_3 + P_2) \\&+ (D_x V)\, (D_x X)^{-1}\, \pa_{v_i}(D_{x} X)
\, (D_x X)^{-1}\, (I + P_2) - (D_x V)\, (D_x X)^{-1}\,\pa_{v_i} P_2, \end{split}\end{equation*}
and arguing in the same way for the spatial derivatives, we obtain
\begin{equation}\label{dvp5}
 \max_{i=1,2,3}\Big(\|\pa_{v_i} P_5\|_{L^{\infty}([0,t]\times\R^3\times\R^3)}+\|\pa_{x_i} P_5\|_{L^{\infty}([0,t]\times\R^3\times\R^3)}\Big) \le 8.
 \end{equation}
\par
Also,
$$P_6=(D_x X)^{-1}(I_3+P_2)(I_3+s P_5)^{-1},$$
therefore
\begin{equation*}
\begin{split}
\pa_{x_i} P_6&=-(D_x X)^{-1}\pa_{x_i}(D_x X) (D_x X)^{-1}(I_3+P_2)(I_3+s P_5)^{-1}\\
&+(D_x X)^{-1}(\pa_{x_i} P_2) (I_3+ s P_5)^{-1}\\
&-(D_x X)^{-1}(D_v X)(I_3+s P_5)^{-1} (\pa_{x_i} P_5 )(I_3+ s P_5)^{-1},
\end{split}
\end{equation*}
where we have used that $s(I_3+P_2)=D_v X.$
Therefore combining the previous estimates (in particular \eqref{nnquatre}, \eqref{2013}, \eqref{nnnun-M}, \eqref{ineq:derivative-2} \eqref{ineq:P2} and \eqref{dvp5}) we obtain
\begin{equation}\label{dxp6}
 \max_{i=1,2,3}\|\pa_{x_i} P_6\|_{L^{\infty}([0,t]\times\R^3\times\R^3)} \le 400.
 \end{equation}

\subsection{The flow and the Duhamel formula}

In this paragraph we establish the formula \eqref{eq:rho2}.
Let $g: \R^3\times \R^3\to \R^3$ be $C^1$. We set
$$h(s,x,v)=g(X(s,x,v),V(s,x,v)).$$  Then
\begin{equation*}
\begin{split}
(D_x h)(s,x,v) &=(D_xg) (X(s,x,v), V(s,x,v)) (D_x X)(s,x,v)+(D_vg) (X(s,x,v), V(s,x,v)) (D_x V)(s,x,v)\\
(D_v h)(s,x,v) &=(D_xg) (X(s,x,v), V(s,x,v)) (D_v X)(s,x,v)+(D_vg) (X(s,x,v), V(s,x,v)) (D_v V)(s,x,v).
\end{split}
\end{equation*}
We next recall that $D_xX$ and $D_vV - (D_x V)\,(D_x X)^{-1} (D_v X)$ are invertible and we set $$M=M(x,v)=(D_v V- D_xV(D_x X)^{-1}D_v X)^{-1}, \quad N=N(x,v)=(D_x X)^{-1}(D_v X) M(x,v).$$
Then, we obtain
\begin{equation*}
\begin{split}
(D_v g)( X(s,x,v), V(s,x,v))&= (D_vh)(x,v)M(x,v)-(D_xh)(x,v)N(x,v)\end{split}\end{equation*}
so that, denoting by $\dev(A)$ the column vector such that $[\dev(A)]_i=\dev(A_i)$ where $A_i$ is the $i$-th line of $A$,
\begin{equation*}
\begin{split}(\dev_v g) (X(s,x,v), V(s,x,v))&=\text{tr}((D_v g)( X(s,x,v), V(s,x,v))
=\text{tr}((D_v h) M)-\text{tr}((D_xh) N)\\
&=\dev_v(Mh)-\dev_v(^tM)\cdot h-\dev_x(Nh)+\dev_x(^tN)\cdot h. 
\end{split}
\end{equation*}
Applying the formula above with $g=(E_{\text{int}}+F_{\text{int}})f(t-s,\cdot,\cdot)$,
 we get after integrating with respect to $v$: 
\begin{equation}\label{nrho}
\begin{split}
\rho_2(t,x)& =\dev_x  \int_0^t\int_{\R^3} N(s,x,v) [(E_{\text{int}}+F_{\text{int}})f](t-s,X(s,x,v),V(s,x,v))\,dv\,ds\\
+ \int_0^t\int_{\R^3}& ( \dev_v (^tM)(x,v)-\dev_x (^tN)(x,v))
[ (E_{\text{int}}+F_{\text{int}})f](t-s,X(s,x,v),V(s,x,v))\,dv\,ds.
\end{split}
\end{equation}

\bigskip

\bigskip

\noindent
{\bf Acknowledgement:} The second author is supported by the ANR projects GEODISP ANR-12-BS01-0015-01  and SCHEQ ANR-12-JS01-0005-01.\\ 
The third author is supported by ERC Grant MAQD 240518 and has been partially supported by GEODISP ANR-12-BS01-0015-01 and by GREFI-MEFI 2011, that are gratefully acknowledged.

\adresse

\end{document}